\documentclass[a4paper,USenglish,cleveref,nameinlink]{lipics-v2021}
\nolinenumbers
\hideLIPIcs
\usepackage[dvipsnames]{xcolor}
\usepackage{tikz}
\usepackage{algorithm}
\usepackage[noend]{algpseudocode}

\algrenewcommand\algorithmicindent{2.4em}%
\usepackage{amssymb,amsmath,mathtools,bm}
\allowdisplaybreaks

\usetikzlibrary{decorations.pathreplacing}

\newcommand{\N}{\mathbb{N}}
\newcommand{\R}{\mathbb{R}}

\makeatletter
\renewcommand\@oddfoot{
	\hfil
	\rlap{%
		\vtop{%
			\vskip10mm
			\colorbox[rgb]{0.99,0.78,0.07}
			{\@tempdima\evensidemargin
				\advance\@tempdima1in
				\advance\@tempdima\hoffset
				\hb@xt@\@tempdima{%
					\textcolor{darkgray}{\normalsize\sffamily
						\bfseries\quad
						\expandafter\textsolittle\expandafter{
							arXiv.org}}%
					\strut\hss}}}}
}
\makeatother

\DeclarePairedDelimiter\abs{\lvert}{\rvert}
\DeclareMathOperator*{\argmin}{arg\;min}

\title{Priority Algorithms with Advice for Disjoint Path Allocation Problems}

\author{Hans-Joachim Böckenhauer}{Department of Computer Science, ETH Zürich, Switzerland}{hjb@inf.ethz.ch}{https://orcid.org/0000-0001-9164-3674}{}
\author{Fabian Frei}{Department of Computer Science, ETH Zürich, Switzerland}{fabian.frei@inf.ethz.ch}{https://orcid.org/0000-0002-1368-3205}{}
\author{Silvan Horvath}{Department of Mathematics, ETH Zürich, Switzerland}{horvaths@student.ethz.ch}{https://orcid.org/0000-0003-2629-5719}{}

\authorrunning{H.-J.\,Böckenhauer, F.\,Frei, S.\,Horvath}

\Copyright{Hans-Joachim Böckenhauer, Fabian Frei, Silvan Horvath}



\begin{document}

\maketitle

\begin{abstract}
We analyze the Disjoint Path Allocation problem (\textnormal{DPA}) in the priority framework. Motivated by the problem of traffic regulation in communication networks, \textnormal{DPA} consists of allocating edge-disjoint paths in a graph. While online algorithms for \textnormal{DPA} have been thoroughly studied in the past, we extend the analysis of this optimization problem by considering the more powerful class of priority algorithms. Like an online algorithm, a priority algorithm receives its input only sequentially and must output irrevocable decisions for individual input items before having seen the input in its entirety. However, in contrast to the online setting, a priority algorithm may choose an order on the set of all possible input items and the actual input is then presented according to this order. A priority algorithm is a natural model for the intuitively well-understood concept of a greedy algorithm.

Apart from analyzing the classical priority setting, we also consider priority algorithms with advice. Originally conceived to study online algorithms from an information-theoretic point of view, the concept of advice has recently been extended to the priority framework.

In this paper, we analyze the classical variant of the \textnormal{DPA} problem on the graph class of paths, the related problem of Length-Weighted \textnormal{DPA}, and finally, \textnormal{DPA} on the graph class of trees. We show asymptotically matching upper and lower bounds on the advice necessary for optimality in \textnormal{LWDPA} and generalize the known optimality result for \textnormal{DPA} on paths to trees with maximal degree at most 3. On trees with maximal degree greater than 3, we prove matching upper and lower bounds on the approximation ratio in the advice-free priority setting. Finally, we present upper and lower bounds on the advice necessary to achieve optimality on such trees.
\end{abstract}

\section{Introduction}\label{sec:introduction}

\subsection{The Model}\label{sec:introprioalg}

\subparagraph*{Priority Algorithms.}
Priority algorithms, a more powerful variant of online algorithms, were originally conceived by Borodin, Nielsen, and Rackoff \cite{borodin2003incremental} with the goal of modeling greedy algorithms. Online algorithms receive their input sequentially in the form of individual \emph{requests}, where an irrevocable decision for each request must be made before the subsequent request is revealed. While this constraint already forces the online algorithm to some degree of greediness, actual greedy algorithms often operate by first preprocessing their input to ensure a certain order among the individual requests. They then operate in an online fashion on the preprocessed data, such that requests are presented in the predetermined order. This idea of exploiting the order in which the input is presented to the algorithm is captured in the priority model.

Like a classical online algorithm, a priority algorithm is presented with its input in a sequential fashion and has to output a decision for each request before being fed the next one. However, unlike in the classical online setting, a priority algorithm imposes a \emph{priority order} (denoted by $\prec$) on the set, or \emph{universe}, of all possible requests. The subset of the universe comprising the actual input is then fed to the priority algorithm according to this order.

We further distinguish \emph{fixed-priority} algorithms from the more general class of \emph{adaptive-priority} algorithms. Both types choose a priority order before being presented with the first request. While this initial priority order persists during the entire operation of a fixed-priority algorithm, an adaptive priority algorithm may change the order after each processed request. For example, Prim's algorithm for finding a minimum spanning tree is a classical greedy algorithm that can only be modeled in the adaptive-priority framework.

Priority algorithms have been studied for a range of optimization problems, such as set cover \cite{setcover}, vertex cover, independent set, and colorability \cite{graphpriocite}, auctions \cite{auctionsprio}, scheduling \cite{borodin2003incremental}, and satisfiability \cite{maxsatprio}.

\subparagraph*{Advice.}
For many optimization problems, one notices that the quality of a solution produced by an online algorithm is much worse than that of an optimal solution. This of course stems from the fact that an online algorithm is oblivious of the remaining input when making its decision for each request, i.e., it lacks information. To analyze this information deficiency, Dobrev et al.~\cite{dobrev2008much} introduced the measure of the \emph{advice complexity} of an online problem, a concept that was refined by Emek et al.~\cite{DBLP:journals/tcs/EmekFKR11}, by Hromkovi{\v{c}} et al.~\cite{hromkovivc2010information}, and by  Böckenhauer et al.~\cite{bockenhauer2017online}, whose model we assume in this paper.

Intuitively, advice complexity is the amount of information about the input an online algorithm needs in order to produce a solution of certain quality. More concretely, an \emph{online algorithm with advice} is given an \emph{advice tape} -- an infinite binary string -- that contains information about the specific input the algorithm has to operate on. The advice tape can be thought of as being provided by an oracle that has unlimited computational power, knows the entire input as well as the inner workings of the algorithm, and cooperates with it. The advice complexity is the number of bits of the advice tape the algorithm accesses during its runtime. An overview of results in the advice framework can be found in the survey by Boyar et al. \cite{DBLP:journals/csur/BoyarFKLM17} and in the textbook by Komm \cite{komm2016introduction}.

The concept of advice has recently been extended to the priority framework by Borodin et al.~\cite{borodin2020advice} and by Boyar et al.~\cite{boyar2019advice}. There are multiple possible models of priority algorithms with advice. While an online algorithm with advice only uses the advice string to make decisions for requests, a priority algorithm must also decide what priority order to use. It is subject to modeling whether or not the priority order should also be allowed to depend on the advice. Boyar et al.~\cite{boyar2019advice} distinguish four possible models with different priority-advice dependencies. However, since both our upper and our lower bounds hold in all four models, we will not assume any specific one of them.

\subsection{Disjoint Path Allocation}\label{sec:introdpa}
The Disjoint Path Allocation problem (\textnormal{DPA}) is a standard online optimization problem. Motivated by the real-world problem of routing calls in communication networks, the \textnormal{DPA} problem consists in allocating edge-disjoint paths between pairs of vertices in a graph. More concretely, the input of an online algorithm for \textnormal{DPA} consists of some graph $G$ in a fixed class of graphs -- we will mainly consider the classes of paths and trees -- followed by a sequence of vertex pairs in $G$, where pairs are revealed one after the other. For each request, i.e., vertex pair $[x,y]$, that is revealed, the online algorithm must immediately decide whether to \emph{accept} or \emph{reject} it. If the algorithm accepts $[x,y]$, it must allocate an $x$-$y$-path in $G$ that does not share an edge with any previously allocated path.\footnote{If the underlying graph $G$ is itself a path, the pair $[x,y]$ can be interpreted as the closed interval between $x$ and $y$ -- thus the use of square brackets. We identify $[x,y]$ with $[y,x]$.} The decision for $[x,y]$ may only depend on previously received vertex pairs and on the graph $G$. The goal is to accept as many vertex pairs as possible.

The classical \textnormal{DPA} problem (also simply referred to as \textnormal{DPA}) considers the problem on the graph class of paths. Generalizing \textnormal{DPA} to trees results in the so-called CAT problem (\emph{Call Admission on Trees}). Note that these versions of the problem allow to somewhat simplify the setting, since there is only one $x$-$y$-path between every vertex pair $[x,y]$ in a tree. Thus, an online or priority algorithm for \textnormal{DPA} on trees and paths must essentially only make a binary decision for each request, namely whether to accept or reject it.

The \textnormal{DPA} problem is well-studied in the online framework with and without advice. Advice-free online \textnormal{DPA} on various graph classes is discussed in Chapter 13 of the textbook by Borodin and El-Yaniv \cite{borodin2005online}. Online algorithms with advice for \textnormal{DPA} on paths are analyzed by Böckenhauer et al.~\cite{bockenhauer2017online}, by Barhum et al.~\cite{barhum2014power}, and discussed in Chapter 7 of the textbook by Komm \cite{komm2016introduction}. Böckenhauer et al.~analyze online algorithms with advice for \textnormal{DPA} on the graph class of trees \cite{bockenhauer2019call} and on the graph class of grids \cite{bockenhauer2018call}.

We will also consider a natural variant of classical \textnormal{DPA} that was introduced by Burjons et al.~\cite{burjons2018length}, namely the Length-Weighted \textnormal{DPA} Problem (\textnormal{LWDPA}). Inspired by the real-world problem of trying to maximize a meeting room's occupancy, the goal in \textnormal{LWDPA} consists not in maximizing the number of accepted paths, but rather their total length.

In \cref{sec:dpapaths}, we analyze the classical \textnormal{DPA} problem on paths. In \cref{sec:LWDPA}, we consider the related problem of Length-Weighted \textnormal{DPA}. In \cref{sec:CAT}, we analyze \textnormal{DPA} on trees, and finally, in \cref{sec:conclusion}, we very briefly discuss \textnormal{DPA} on the class of grids.

\subsection{Preliminaries}\label{sec:preliminaries}
There is a certain subtlety that must be considered when applying the priority framework to the online problem \textnormal{DPA}. Namely, we do not want what we informally consider to be the first input item a priority algorithm for \textnormal{DPA} receives -- the underlying graph $G$ on which the given instance is defined -- to be a request to which a priority can be assigned. We thus define requests of the priority problem \textnormal{DPA} to be of the form $[x,y]^G$, where $G=(V,E)$ is some graph, $x \neq y \in V$, and we identify $[x,y]^G$ with $[y,x]^G$. Hence, the underlying graph does not need to be communicated to the algorithm separately. Informally, we will still think of $G$ as being the first input item that is presented, even though we technically assume that it is only revealed together with the first request.

As mentioned previously, we typically distinguish different variants of the \textnormal{DPA} problem with respect to the graph class from which the underlying graph is chosen, e.g., the class of paths or the class of trees. Thus, we define the universe of all possible requests to be the set $\mathcal{U}\,:=\{\,[x,y]^G\;|\;G \in \mathcal{G},\;x \neq y \in V(G)\,\}$, where $\mathcal{G}$ is some arbitrary graph class. An \emph{instance} of \textnormal{DPA} on $\mathcal{G}$ is then simply a subset $I \subseteq \mathcal{U}$ with $[x,y]^G,\,[v,w]^H \in I\,\Rightarrow \,G=H$. Note that the same request cannot appear twice in an instance.

As explained above, a priority algorithm for \textnormal{DPA} chooses a priority order $\prec$ on $\mathcal{U}$ and an instance $I$ is then presented to the algorithm in accordance with this order -- such that $\max_{\prec}I$ is presented first. A fixed-priority algorithm chooses a fixed order, while an adaptive-priority algorithm chooses a new order after each request. Note that we assume that $\prec$ is total in order to simplify notation -- but we obviously never care about the order relation between requests defined on distinct graphs, because such requests never appear in the same instance. In fact, we often do not even require the order on the subsets of $\mathcal{U}$ that correspond to sets of requests belonging to the same graph to be total. Since, by the Szpilrajn extension theorem \cite{szpilrajn1930extension}, every partial order on a set can be extended to a total order, we will thus only define partial orders on $\mathcal{U}$.

On a further note, as explained earlier, we may simplify the setting if every graph $G\in \mathcal{G}$ is cycle-free. Since paths in such a graph are uniquely determined by their end points, we can identify the request $[x,y]^G$ with the unique $x$-$y$-path in $G$. The algorithm must then only decide whether to accept or reject the path $[x,y]^G$. Furthermore, this identification allows us to use certain helpful terminology -- such as that two requests $[x,y]^G$ and $[v,w]^G$ \emph{intersect} if the corresponding paths edge-intersect. If $[x,y]^G$ cannot be accepted because it intersects the previously accepted request $[v,w]^G$, we say that $[v,w]^G$ \emph{blocks} $[x,y]^G$.

For a priority algorithm \textsc{ALG} and an instance $I$, we define $\textsc{ALG}(I) \subseteq I$ to be the set of requests accepted by \textsc{ALG} and $\textsc{OPT}(I) \subseteq I$ to be an optimal solution for $I$. If $\textsc{ALG}$ uses the advice string $\Phi$, we write $\textsc{ALG}^{\Phi}(I)$. For the \textnormal{DPA} problem, we denote by $\abs{\textsc{ALG}(I)}$ and $\abs{\textsc{OPT}(I)}$ the sizes of the corresponding sets, i.e., the number of requests accepted by \textsc{ALG} and the optimal solution, respectively. For the related problem of \textnormal{LWDPA} -- where the goal consists in maximizing the total length of the accepted paths -- we denote that quantity by $\abs{\textsc{ALG}(I)}$ or $\abs{\textsc{OPT}(I)}$. We refer to $\abs{\textsc{ALG}(I)}$ and $\abs{\textsc{OPT}(I)}$ as the \emph{gains} on $I$ of \textsc{ALG} and the optimal solution, respectively.

We say that \textsc{ALG} is \emph{strictly $c$-competitive} for some $c:\;\mathcal{G} \to \R_{\ge1}$ if, for every $G \in \mathcal{G}$ and every instance $I$ defined on $G$,
$\abs{\textsc{OPT}(I)}\,/\,\abs{\textsc{ALG}(I)}\leq c(G)$.
The \emph{strict competitive ratio} or \emph{approximation ratio} of \textsc{ALG} is the pointwise infimum over all $c$ for which \textsc{ALG} is strictly $c$-competitive. \textsc{ALG} is optimal if it has a strict competitive ratio of $1$.

As for priority algorithms with advice, we say that such an algorithm is strictly $c$-competitive with advice complexity $b:\;\mathcal{G} \to \N$ if, for every $G\in \mathcal{G}$ and every instance $I$ defined on $G$, there exists an advice string $\Phi$ such that $\abs{\textsc{OPT}(I)}\,/\,\abs{\textsc{ALG}^{\Phi}(I)}\leq c(G)$ and each of the algorithm's decisions only depends on the first $b(G)$ bits of $\Phi$.

Note that often, the study of online algorithms is centered around the analysis of their \emph{non-strict} competitive ratio (often just called competitive ratio), where \textsc{ALG} is (non-strictly) $c$-competitive for some $c \in \R_{\ge1}$ if there exists a constant $\alpha \in \R$ such that, for every instance $I$, $\abs{\textsc{OPT}(I)} \leq c \cdot \abs{\textsc{ALG}(I)} + \alpha$. The additive constant $\alpha$ allows \textsc{ALG} to be c-competitive even if there exist finitely many instances $I$ for which $\abs{\textsc{OPT}(I)}\,/\,\abs{\textsc{ALG}(I)} > c$. However, since there are only finitely many instances of \textnormal{DPA} that are defined on the same underlying graph and since we want to analyze competitiveness with respect to that graph, we focus on analyzing the strict competitive ratio. Note however that all lower bounds on the strict competitive ratio that are presented in this paper can easily be translated into lower bounds on the non-strict competitive ratio if we dispense with letting the latter depend on the underlying graph.\footnote{This can be done by embedding multiple copies of the construction used in the respective proof in a suitable graph.}
\newpage
\section{DPA on Paths}\label{sec:dpapaths}

We begin with the classical variant of the \textnormal{DPA} problem -- namely \textnormal{DPA} on the graph class of paths, $\mathcal{P}$. DPA on paths is essentially a discrete variant of the interval scheduling problem with unit profit on a single machine -- an optimization problem that was already studied in the priority framework by Borodin et al.~\cite{borodin2003incremental}. Both problems are solved optimally by a fixed-priority algorithm without advice, a fact independently shown by Carlisle and Lloyd \cite{DBLP:journals/dam/CarlisleL95} and by Faigle and Nawijn \cite{DBLP:journals/dam/FaigleN95}. Below, we give a proof that is specific to the \textnormal{DPA} problem.

\begin{theorem}[{Carlisle and Lloyd \cite{DBLP:journals/dam/CarlisleL95} and Faigle and Nawijn \cite{DBLP:journals/dam/FaigleN95}}]\label{thm:optdpa}
There is an optimal fixed-priority algorithm without advice for \textnormal{DPA} on paths.
\end{theorem}

\begin{proof}
For a path $P\in \mathcal{P}$ of length $l$, we assume that its vertex set is $V(P)=\{0,\,1,\,2,\,...\,,\,l\}$, with an edge between vertices $i$ and $i+1,\;i\in V(P)\setminus \{l\}$. We say that the vertex $0$ is the left end and the vertex $l$ the right end of $P$.

Let $\prec$ order requests according to their right end points -- such that priority increases the further to the left the right end point is. More formally, define for each $P\in \mathcal{P},\;x<y$ and $x'<y' \in V(P)$:
\[ y<y' \implies [x,\,y]^P \succ [x',\,y']^P \]
and extend this to some total order on $\mathcal{U}$. With this priority order, consider the simple priority algorithm \textsc{GREEDY} that accepts a request whenever it is not already blocked by a previously accepted request. It follows by induction on the number of requests of an instance that \textsc{GREEDY} is optimal.

On instances consisting of only one request, this is trivial. Thus, let $I$ be an instance with at least two requests and say that $I$ is defined on the path $P$.

Let $p \in {\textsc{GREEDY}(I)}$ be the last request in $I$ that is accepted by \textsc{GREEDY}, i.e., $p:=\min_{\prec}{\textsc{GREEDY}(I)}$. Consider the instance $I'$ consisting of the same requests as $I$, except that $p$ and all requests with lower priority than $p$ that intersect it are omitted. Since all requests following $p$ in $I$ are rejected by \textsc{GREEDY}, they must intersect previously accepted requests. In particular, requests with lower priority than $p$ that also belong to the instance $I'$ must intersect accepted requests preceding $p$, since they do not intersect $p$ itself by definition of $I'$. Therefore, they are also rejected by \textsc{GREEDY} when processing the instance $I'$, because the two instances do not differ on requests preceding $p$ and thus \textsc{GREEDY}'s decisions on these requests do neither, i.e., the requests blocking paths in $I \cap I'$ with lower priority than $p$ are accepted both on $I$ as well as on $I'$. It follows that $|\textsc{GREEDY}(I)|=|\textsc{GREEDY}(I')|+1$.

Write $p=[p_1,\,p_2]^P,\;p_1<p_2 \in V(P)$, and consider now the requests in $I$ that are not included in $I'$, i.e., those that follow and intersect $p$, including $p$ itself. By definition of the priority order, for all such requests $[x,\,y]^P$ it holds that $y\geq p_2$. Since $[x,\,y]^P$ intersects $p$, we have $x<p_2$. Combining these two inequalities shows that the edge $\{p_2-1,\,p_2\}$ of $P$ is always contained in $[x,\,y]^P$. Thus, at most one of the requests not included in $I'$ can be part of the optimal solution $\textsc{OPT}(I)$ of $I$, since two such requests would intersect on the edge $\{p_2-1,\,p_2\}$. This yields $|\textsc{OPT}(I)|\leq|\textsc{OPT}(I')|+1$. Applying the induction hypothesis to $I'$ concludes the proof:
\[\frac{\abs{\textsc{OPT}(I)}}{\abs{\textsc{GREEDY}(I)}} \leq \frac{\abs{\textsc{OPT}(I')}+1}{|\textsc{GREEDY}(I')|+1} = \frac{\abs{\textsc{GREEDY}(I')}+1}{\abs{\textsc{GREEDY}(I')}+1}=1.\qedhere\]
\end{proof}

In contrast to this optimality result, there is a tight lower bound of $l$ -- the length of the underlying path graph -- on the strict competitive ratio for \textnormal{DPA} on paths in the standard online setting, as seen in the textbook by Borodin and El-Yaniv \cite{borodin2005online}.

As we will see in Section \ref{sec:CAT}, the above result generalizes. Namely, the \textnormal{DPA} problem on trees with maximum degree at most 3 -- which paths obviously are -- is also solved optimally by a fixed-priority greedy algorithm requiring no advice.

However, before considering the more general problem of \textnormal{DPA} on trees, we consider the Length-Weighted \textnormal{DPA} problem (\textnormal{LWDPA}). Besides being an interesting problem in the priority-with-advice framework in its own right, the analysis of \textnormal{LWDPA} also serves to introduce a few techniques that will be important in the analysis of \textnormal{DPA} on trees.

\section{Length-Weighted DPA}\label{sec:LWDPA}

\textnormal{LWDPA} is conceptually similar to the interval scheduling problem with proportional profits on a single machine. This optimization problem was also studied in the advice-free priority setting by Borodin et al.~\cite{borodin2003incremental}. The tight bound of 3 on the approximation ratio for the latter problem corresponds to an upper bound of $(3-3/l)$ and a lower bound of $(3-\mathcal{O}(l^{-1/3}))$ for \textnormal{LWDPA}. The proofs of the following two results are based on constructions by Borodin et al.~\cite{borodin2003incremental}.

\begin{theorem}\label{thm:uplwdpa}
There is a strictly $(3-3/l)$-competitive fixed-priority algorithm without advice for \textnormal{LWDPA}.
\end{theorem}

\begin{proof}
In accordance with the notation $\abs{\textsc{ALG}(I)}$ for the total length of the paths in $\textsc{ALG}(I)$, we denote by $\abs{[x,\,y]^P}$ the length of the unique $x$-$y$-path in $P$.

We use a priority order that orders requests primarily by non-increasing length and secondarily by their position on the path, the leftmost taking precedence. More formally, define, for each $P\in \mathcal{P},\;x<y$ and $x'<y' \in V(P)$: 
\[ (y-x>y'-x')\;\textnormal{or}\; (y-x=y'-x'\;\textnormal{and}\;x<x') \implies [x,\,y]^{P} \succ [x',y']^{P}.  \]
With this priority order, we again consider the fixed-priority algorithm \textsc{GREEDY} that accepts every request it can. We use induction on the size of an instance.

On instances containing only one request, the result holds trivially since \textsc{GREEDY} is optimal. Thus consider an instance $I$ consisting of at least two requests. Let $p \in {\textsc{GREEDY}(I)}$ be the last request in $I$ that is accepted by \textsc{GREEDY}. Denote further by $L$ the length of the subpath of $P$ constructed as the union of those paths of $I$ which follow and intersect $p$ (including $p$ itself). Define the instance $I'$ again as $I$ minus $p$ and all requests following and intersecting it (thus requests of $I$ which, when combined, yield a path of length $L$ are omitted in $I'$).

Since all requests following $p$ in $I$ are rejected by \textsc{GREEDY}, they must intersect previously accepted requests. In particular, requests with lower priority than $p$ that also belong to $I'$ must intersect accepted requests preceding $p$ since they do not intersect $p$ itself by definition of $I'$. Therefore, they are also rejected by \textsc{GREEDY} when processing the instance $I'$, because the two instances do not differ on requests preceding $p$ and thus the algorithm's decisions on them do neither. It follows that $|\textsc{GREEDY}(I)|=|\textsc{GREEDY}(I')|+|p|$.

Since an offline algorithm processing $I'$ can simply accept those requests that are also part of a fixed optimal solution for $I$, we have the simple lower bound $|\textsc{OPT}(I')|\geq|\textsc{OPT}(I)|-L$. It remains to bound $L$.

\begin{figure}[t]
\centering
\resizebox{\textwidth}{!}{
\begin{tikzpicture}
\node[circle,fill=black, inner sep=0, minimum size=4] at (0,0) (node1) {};
\foreach \n in {2,...,19}{
		\pgfmathtruncatemacro{\nminusone}{\n - 1}
       \node[circle,fill=black, inner sep=0, minimum size=4] at (\nminusone,0) (node\n) {};
       \draw[line width=1] (node\nminusone)--(node\n);
    }
\node[circle, inner sep=0, minimum size=3] at (6,-0.5) (p1) {};
\node [circle, inner sep=0, minimum size=3] at (11,-0.5) (p2) {};
\draw [line width=1.5] (p1)--(p2);

\node[circle] at (8.5,-1) (label) {$p$};
\node[circle] at (5,0.5) (label) {$q_1$};
\node[circle] at (8.5,0.5) (label) {$q_2$};
\node[circle] at (12.5,0.5) (label) {$q_3$};

\node [circle, inner sep=0, minimum size=3] at (3,0.95) (a1) {};
\node [circle, inner sep=0, minimum size=3] at (7,0.95) (a2) {};
\draw [line width=1.5] (a1)--(a2);

\draw [line width=1, decorate,
    decoration = {brace}] (3, 1.35) --  (7, 1.35);

\node[circle] at (5, 1.85) (label) {$\leq \abs{p}-1$};

\node [circle, inner sep=0, minimum size=3] at (7,0.95) (b1) {};
\node [circle, inner sep=0, minimum size=3] at (10,0.95) (b2) {};
\draw [line width=1.5] (b1)--(b2);

\node [circle, inner sep=0, minimum size=3] at (10,0.95) (c1) {};
\node [circle, inner sep=0, minimum size=3] at (15,0.95) (c2) {};
\draw [line width=1.5] (c1)--(c2);

\draw [line width=1, decorate,
    decoration = {brace}] (10,1.35) --  (15,1.35);

\node[circle] at (12.5, 1.85) (label) {$\leq \abs{p}$};

\end{tikzpicture}

}

\caption{Example of three requests $q_1$, $q_2$ and $q_3$ following and intersecting the path $p$ that attain the maximal combined length $3|p|-3$. Note that $q_1$ must have length at most $\abs{p}-1$ and $q_3$ must have length at most $\abs{p}$ since otherwise, they would precede $p$ by definition of the priority order.}
\label{figuplwdpa}
\end{figure}
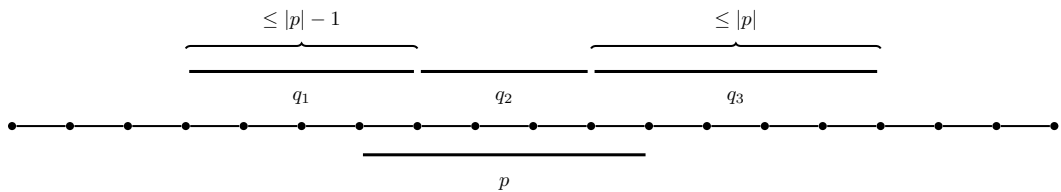

By definition of the priority order, all requests having lower priority than $p$ have length at most $\abs{p}$. To be more precise, requests following and intersecting $p$ either have their left end point to the left of $p$'s left end point but length at most $|p|-1$, or their left end point is contained in $p$, in which case they can have length at most $|p|$ (see \cref{figuplwdpa}). Thus it is easy to see that $L\leq3|p|-3$.

Using the inequalities above, we get \[\frac{|\textsc{OPT}(I)|}{|\textsc{GREEDY}(I)|}\leq \frac{|\textsc{OPT}(I')|+3|p|-3}{|\textsc{GREEDY}(I')|+|p|}\]
Applying the induction hypothesis to $I'$ and using the fact that $\abs{p}\leq l$ yields the desired result:
\begin{align*}
\frac{|\textsc{OPT}(I')|+ 3|p|-3}{|\textsc{GREEDY}(I')|+|p|}
& \leq\frac{(3-3/l)\;|\textsc{GREEDY}(I')|+|p|\;(3-3/|p|)}{|\textsc{GREEDY}(I')|+|p|}\\
&\leq \frac{(3-3/l)\;|\textsc{GREEDY}(I')|+|p|\;(3-3/l)}{|\textsc{GREEDY}(I')|+|p|}\\
&=3-3/l.\qedhere	
\end{align*}
\end{proof}

\begin{theorem}\label{lowlwdpa}
The approximation ratio of any adaptive-priority algorithm without advice for \textnormal{LWDPA} is at least $(3-\mathcal{O}(l^{-1/3}))$.
\end{theorem}

\begin{proof}
Let \textsc{ALG} be an adaptive-priority algorithm for \textnormal{LWDPA} reading no advice and, for $a,b\in \mathbb N_{\geq 3}$, define $P_{a,b}$ to be the path of length $l:=ab^2-2b+2$. This definition results from the following construction (see \cref{figlowlwdpa}): Let $p_1$ be the path $[0,a]^{P_{a,b}}$ and define, for $2\leq i \leq b$, the request $p_i$ to be the path of length $ia$ intersecting the previous path $p_{i-1}$ only on its rightmost edge. Mirror this construction for the paths $p_i$ with $b+1\leq i \leq 2b-1$, i.e., let $p_i$ here be the path of length $(2b-i)a$, also intersecting the previous path on a single edge.\footnote{Our path graph is just large enough to accommodate this construction with the above definition of $l$.} Additionally define, for every edge $e$ in the path $P_{a,b}$, the request $p_e$ to be the corresponding length-one path.

Depending on which of these paths has highest initial priority, we distinguish four cases. In every case, we consider some suitable instance defined on the graph $P_{a,b}$ on which \textsc{ALG} will be bad compared to an optimal solution.
\begin{description}
\item[\textit{Case~1.}] If, among the requests defined above, one of the requests $p_i$ for some $2\leq i\leq b-1$ has highest initial priority, consider the instance $I$ consisting of $p_i$, followed by $p_{i-1}$, $p_{i+1}$, and all length-one paths contained in $p_i\setminus (p_{i-1}\cup p_{i+1})$. If \textsc{ALG} rejects $p_i$, it has no approximation ratio at all since it has a gain of $0$ on the instance consisting only of $p_i$, while the optimal solution has a gain of $1$ on this instance. Thus, \textsc{ALG} must accept $p_i$ and must thus reject the remaining requests in $I$ since they intersect $p_i$. This yields
\begin{align*}
\frac{|\textsc{OPT}(I)|}{|\textsc{ALG}(I)|}&=\frac{|p_{i-1}|+|p_{i}|+|p_{i+1}|-2}{|p_i|}\\
&= \frac{(i-1)a+ia+(i+1)a-2}{ia}\\
&=3-\frac{2}{ia}\\
&\geq 3-\frac{1}{a}.
\end{align*}
The case where $p_i$ with $b+1\leq i\leq 2b-2$ has highest initial priority is clearly symmetric.

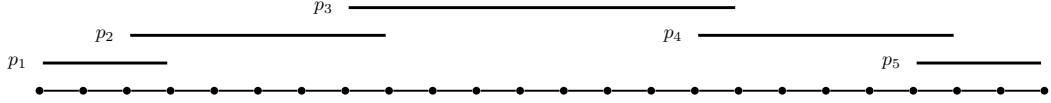
\begin{figure}[t]
\centering
\resizebox{\textwidth}{!}{
\begin{tikzpicture}

\node[circle,fill=black, inner sep=0, minimum size=4] at (0,0) (node1) {};
\foreach \n in {2,...,24}{
		\pgfmathtruncatemacro{\nminusone}{\n - 1}
       \node[circle,fill=black, inner sep=0, minimum size=4] at (\nminusone *18/23,0) (node\n) {};
       \draw[line width=1] (node\nminusone)--(node\n);
    }

\node [circle, inner sep=0, minimum size=3] at (0 *18/23,0.5) (p11) {};
\node [circle, inner sep=0, minimum size=3] at (3 *18/23,0.5) (p12) {};
\draw [line width=1.5] (p11)--(p12);

\node[circle] at (-0.5 *18/23,0.5) (label1) {$p_1$};

\node [circle, inner sep=0, minimum size=3] at (2 *18/23,1) (p21) {};
\node [circle, inner sep=0, minimum size=3] at (8 *18/23,1) (p22) {};
\draw [line width=1.5] (p21)--(p22);

\node[circle] at (1.5 *18/23,1) (label2) {$p_2$};

\node [circle, inner sep=0, minimum size=3] at (7 *18/23,1.5) (p31) {};
\node [circle, inner sep=0, minimum size=3] at (16 *18/23,1.5) (p32) {};
\draw [line width=1.5] (p31)--(p32);

\node[circle] at (6.5 *18/23,1.5) (label3) {$p_3$};

\node [circle, inner sep=0, minimum size=3] at (15 *18/23, 1) (p41) {};
\node [circle, inner sep=0, minimum size=3] at (21 *18/23, 1) (p42) {};
\draw [line width=1.5] (p41)--(p42);

\node[circle] at (14.5 *18/23,1) (label4) {$p_4$};

\node [circle, inner sep=0, minimum size=3] at (20 *18/23,0.5) (p51) {};
\node [circle, inner sep=0, minimum size=3] at (23 *18/23,0.5) (p52) {};
\draw [line width=1.5] (p51)--(p52);

\node[circle] at (19.5 *18/23,0.5) (label5) {$p_5$};

\end{tikzpicture}

}

\caption{The construction used in the proof of \cref{lowlwdpa} for $a=b=3$. In addition to the paths $p_1,...,p_5$ depicted above, the proof uses, for each edge $e$ in the graph, the length-one path $p_e$.}
\label{figlowlwdpa}
\end{figure}

\item[\textit{Case~2.}] If $p_b$ has highest initial priority, let $I$ consist of $p_b$, followed by $p_{b-1}$, $p_{b+1}$ and the length-one requests contained in $p_b\setminus (p_{b-1}\cup p_{b+1})$. Again, \textsc{ALG} must accept the first request $p_b$ and must thus reject the other requests, which yields
\begin{align*}
\frac{|\textsc{OPT}(I)|}{|\textsc{ALG}(I)|}&=\frac{|p_{b-1}|+|p_{b}|+|p_{b+1}|-2}{|p_b|}\\
&= \frac{(b-1)a+ba+(b-1)a-2}{ba}\\
&=3-\frac{2(a+1)}{ba}.\\
\end{align*}
Setting $b:=2(a+1)$ yields the same $(3-1/a)$ lower bound as in the first case.

\item[\textit{Case~3.}] If $p_1$ has highest initial priority, let $I$ be the instance consisting of $p_1$, followed by $p_2$ and all length-one paths contained in $p_1\setminus p_2$. As in the previous cases, \textsc{ALG} must accept the first request $p_1$ and must thus reject the remaining requests, which yields that
\[\frac{|\textsc{OPT}(I)|}{|\textsc{ALG}(I)|}=\frac{|p_1|+|p_2|-1}{|p_1|}= \frac{a+2a-1}{a}=3-\frac{1}{a}.\]
The case where $p_1$'s mirror image $p_{2b-1}$ has highest initial priority is clearly analogous.

\item[\textit{Case~4.}] If one of the paths $p_e$ 	for some edge $e$ has highest initial priority, let $I$ consist of $p_e$, followed by one of the requests $p_i$ intersecting $p_e$. \textsc{ALG} must accept $p_e$ and reject $p_i$. Since $p_i$ has length at least $a\geq 3$, the optimal solution on the other hand consists in rejecting $p_e$ and accepting $p_i$, which 	gives \[\frac{|\textsc{OPT}(I)|}{|\textsc{ALG}(I)|}\geq \frac{a}{1}\geq 	3.\]
\end{description}

Plugging $b=2(a+1)$ into the definition of $l$ yields $1/a \in \mathcal{O}(l^{-1/3})$ and thus the desired result for all $l$ of the form $l= \abs{P_{a,2(a+1)}}$ for some $a \in \mathbb N_{\geq 3}$. For general $l$ of sufficient size, we use the above construction on the subpath $P_{a,2(a+1)}$ of the path of length $l$ for the largest possible $a \in \mathbb N_{\geq 3}$. For this $a$, we have that $\abs{P_{a+1,2(a+2)}} > l $ and thus $1/(a+1) \in \mathcal{O}(\abs{P_{a+1,2(a+2)}}^{-1/3}) \subseteq \mathcal{O}(l^{-1/3})$, which again implies $1/a \in \mathcal{O}(l^{-1/3})$.
\end{proof}

This asymptotically tight bound stands in contrast to a lower bound of $l$ in the standard online framework proved by Burjons et al.~\cite{burjons2018length}.

We now turn our attention to priority algorithms with advice for \textnormal{LWDPA}. First, we present an optimal fixed-priority algorithm with advice. We use a similar strategy to the one used by Burjons et al.~\cite{burjons2018length} for their optimal online algorithm with advice for \textnormal{LWDPA}. The strategy is to fix an optimal solution for the given instance and to encode this optimal solution in the advice string. Burjons et al.~do this by conveying to the algorithm for every vertex in the graph whether or not this vertex is a \emph{transition point}, i.e., a start point or end point of a path contained in the fixed optimal solution. The algorithm then simply checks for each path it is presented whether or not its start point and end point are consecutive transition points and accepts only if they are. The algorithm thus reproduces the fixed optimal solution.

Using this strategy yields that $(l-1)$ advice bits are sufficient for an online algorithm to be optimal since whether or not a vertex is a transition point needs to be conveyed for all but the first and the last vertex in the graph. However, if we instead consider this strategy in the more powerful priority framework, we find that fewer advice bits are sufficient. We require a preliminary definition.

\begin{definition}\label{greedopt}
	Let $I$ be an instance of \textnormal{DPA} or of \textnormal{LWDPA} and $\prec$ a fixed-priority order for the problem at hand. Write $I=\{r_1,\,r_2,\,...\,,\,r_n\}$, where $r_1 \succ r_2 \succ ... \succ r_n$. Inductively define, for $1 \leq k \leq n$ and $S_0:=\emptyset$,
\[
S_k := \begin{cases}
       S_{k-1} \cup \{ r_k \} & \parbox{.7\textwidth}{\text{if there is an} $R \subseteq \{r_{k+1},\,r_{k+2},\,...\,,\,r_n\}$ such that $S_{k-1} \cup \{r_k\} \cup R$ is an optimal solution for $I$, and}\\
        S_{k-1} & \text{otherwise.}
        \end{cases}  
\]
Define the \emph{greediest optimal solution} of $I$ to be 
$\textup{\textsc{OPT}}_{\textnormal{Gr}}(I):= S_n$. 
\end{definition}

It is clear that $\textsc{OPT}_{\textnormal{Gr}}(I)$ is in fact an optimal solution for $I$. More informally, it is the solution produced by an optimal offline algorithm that first computes all optimal solutions for a given instance and then operates on it in a sequential fashion (according to the priority order), accepting a request whenever possible, i.e., never rejecting a request if it could also arrive at an optimal solution by accepting that request. 

\begin{theorem}\label{uplwdpaadv}
There is an optimal fixed-priority algorithm for \textnormal{LWDPA} reading $3\lceil{l/4}\rceil$ advice bits.
\end{theorem}

\begin{proof}
The strategy is to convey the greediest optimal solution of the instance at hand to the algorithm -- which will again use the priority order defined the proof of \cref{thm:uplwdpa}.

Let $I$ be an instance of \textnormal{LWDPA} and $\textsc{OPT}_{\textnormal{Gr}}(I)$ the greediest optimal solution for $I$, as defined in \cref{greedopt}. We communicate to the algorithm the starting points of requests with length at least $2$ in $\textsc{OPT}_{\textnormal{Gr}}(I)$. We use a quasi-greedy algorithm that always accepts the first request that fits the conveyed starting point configuration. Furthermore, our algorithm greedily accepts unblocked paths of length $1$.

Assuming $\textsc{OPT}_{\textnormal{Gr}}(I)$ does not contain a length-one request, the claim is trivial: Since the algorithm receives long requests first, accepting the first request that starts at a designated starting point but does not internally contain the subsequent starting point is consistent with the optimal solution. Requests of length $1$, being the last requests to be presented to the algorithm by definition of the priority order -- must then all be blocked by previously accepted requests, i.e., requests in $\textsc{OPT}_{\textnormal{Gr}}(I)$, since otherwise $\textsc{OPT}_{\textnormal{Gr}}(I)$ could be extended to a better solution.

\begin{figure}[t]
\resizebox{\textwidth}{!}{
\begin{tikzpicture}

\node[circle,fill=black, inner sep=0, minimum size=4] at (0,0) (node1){};

\foreach \n in {2,...,27}{
		\pgfmathtruncatemacro{\nminusone}{\n - 1}
      \node[circle,fill=black, inner sep=0, minimum size=4] at (\nminusone*18/26,0) (node\n){};
       \draw[line width=1] (node\nminusone)--(node\n);
    }

\node[inner sep=0] at (0.5*2*18/26, 0.5) (p1) {};
\node[inner sep=0] at (3 *2 *18/26, 0.5) (p2) {};
\draw [line width=2] (p1)--(p2);

\node[inner sep=0] at (3.5 *2 *18/26, 0.7) (p3) {};
\node[inner sep=0] at (4.5 *2 *18/26, 0.7) (p4) {};
\draw [line width=2] (p3)--(p4);

\node[inner sep=0] at (4.5 *2 *18/26, 0.5) (b1) {};
\node[inner sep=0] at (5.5 *2 *18/26, 0.5) (b2) {};
\draw [line width=2] (b1)--(b2);

\node[inner sep=0] at (5.5 *2 *18/26, 0.7) (c1) {};
\node[inner sep=0] at (6 *2 *18/26, 0.7) (c2) {};
\draw [line width=2] (c1)--(c2);

\node[inner sep=0] at (6.5 *2 *18/26, 0.5) (d1) {};
\node[inner sep=0] at (7.5 *2 *18/26, 0.5) (d2) {};
\draw [line width=2] (d1)--(d2);

\node[inner sep=0] at (7.5 *2*18/26, 0.7) (e1) {};
\node[inner sep=0] at (8 *2 *18/26, 0.7) (e2) {};
\draw [line width=2] (e1)--(e2);

\node[inner sep=0] at (8 *2 *18/26, 0.5) (f1) {};
\node[inner sep=0] at (10 *2 *18/26, 0.5) (f2) {};
\draw [line width=2] (f1)--(f2);

\node[inner sep=0] at (10 *2 *18/26, 0.7) (g1) {};
\node[inner sep=0] at (11 *2 *18/26, 0.7) (g2) {};
\draw [line width=2] (g1)--(g2);

\node[inner sep=0] at (11.5 *2 *18/26, 0.5) (add1) {};
\node[inner sep=0] at (13 *2 *18/26, 0.5) (add2) {};
\draw [line width=2] (add1)--(add2);

\node[inner sep=0] at (0.5 *2 *18/26, 1.1) (h1) {};
\node[inner sep=0] at (2 *2 *18/26, 1.1) (h2) {};
\draw [line width=1] (h1)--(h2);

\node[inner sep=0] at (2 *2 *18/26, 1.3) (i1) {};
\node[inner sep=0] at (2.5 *2 *18/26, 1.3) (i2) {};
\draw [line width=1] (i1)--(i2);

\node[inner sep=0] at (2.5 *2 *18/26, 1.1) (j1) {};
\node[inner sep=0] at (3 *2 *18/26, 1.1) (j2) {};
\draw [line width=1] (j1)--(j2);

\node[inner sep=0] at (4 *2 *18/26, 1.3) (k1) {};
\node[inner sep=0] at (6 *2 *18/26, 1.3) (k2) {};
\draw [line width=1] (k1)--(k2);

\node[inner sep=0] at (7 *2 *18/26, 1.1) (l1) {};
\node[inner sep=0] at (9 *2 *18/26, 1.1) (l2) {};
\draw [line width=1] (l1)--(l2);

\node[inner sep=0] at (12 *2 *18/26, 1.3) (m1) {};
\node[inner sep=0] at (13 *2 *18/26, 1.3) (m2) {};
\draw [line width=1] (m1)--(m2);

\foreach \n in {1,...,6}{
      \node at (\n*4*18/26-0.5*18/26,0.4) (node1\n){};
      \node at (\n*4*18/26-0.5*18/26,-0.4) (node2\n){};
       \draw[line width=1, dashed] (node1\n)--(node2\n);
    }
    
\draw[-stealth, line width=1] (1 *18/26,-1)--(1 *18/26,-0.5);
\draw[-stealth, line width=1] (7 *18/26,-1)--(7 *18/26,-0.5);
\draw[-stealth, line width=1] (9 *18/26,-1)--(9 *18/26,-0.5);
\draw[-stealth, line width=1] (13 *18/26,-1)--(13 *18/26,-0.5);
\draw[-stealth, line width=1] (16 *18/26,-1)--(16 *18/26,-0.5);
\draw[-stealth, line width=1] (20 *18/26,-1)--(20 *18/26,-0.5);
\draw[-stealth, line width=1] (23 *18/26,-1)--(23 *18/26,-0.5);

\end{tikzpicture}
}

  \caption{Example of an instance of \textnormal{LWDPA} with the corresponding advice as in the proof of \cref{uplwdpaadv}. Bold lines denote requests that are included in the greediest optimal solution and arrows indicate the starting positions that are conveyed by the advice string. Note that rejecting the leftmost bold request and in turn accepting the three requests intersecting it would yield a different optimal solution which would however not be the greediest.}\label{figuplwdpaadv}
\end{figure}
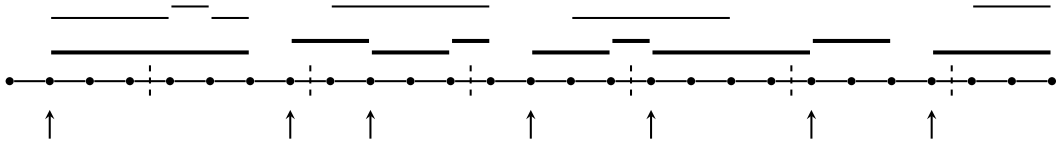

If $\textsc{OPT}_{\textnormal{Gr}}(I)$ does contain length-one requests, our quasi-greedy algorithm might a priori be unable to accept such requests $e_1, e_2,...,e_k \in \textsc{OPT}_{\textnormal{Gr}}(I) $ because it mistakenly accepted an earlier request $q$ blocking $e_1, e_2,...,e_k$. However, this is evidently impossible: Since $e_1, e_2,...,e_k$ have length $1$ and intersect $q$, they must all be contained in $q$. This, however, means that accepting $q$ and rejecting $e_1, e_2,...,e_k$ yields a solution that is at least as good as, but greedier than the fixed optimal solution, which contradicts the fact that the fixed optimal solution is the greediest one. Our quasi-greedy algorithm thus agrees with the fixed optimal solution on all requests having length at least $2$. Length-one requests can then be greedily accepted by the algorithm in the end, unless they are already blocked -- in which case they are not part of $\textsc{OPT}_{\textnormal{Gr}}(I)$.

Thus, we divide the underlying path $P$ of length $l$ into $\lceil{l/4}\rceil$ subpaths $p_1,\,p_2,\,...\,,\,p_{\lceil{l/4}\rceil}$ of length 4 each -- shortening the rightmost subpath if necessary (see \cref{figuplwdpaadv}). We encode all starting points belonging to $p_i$ (where we consider the last vertex of $p_i$ to belong to $p_{i+1}$) of requests contained in $\textsc{OPT}_{\textnormal{Gr}}(I)$ of length at least 2 using 3 bits. This is possible because there are $2^3=8$ configurations of such starting points in $p_i$: 4 possible positions of a single starting point, plus 3 possible configurations of 2 distinct starting points (considering that they must have a distance of at least 2), plus the possibility of $p_i$ not containing any such starting point. Concatenating these 3-bit strings results in an advice string of length $3\lceil{l/4}\rceil$.
\end{proof}

With this upper bound on the amount of advice necessary to achieve optimality, we continue with a corresponding lower bound. However, instead of simply giving a lower bound for optimality, the following theorem bounds the advice complexity of priority algorithms for \textnormal{LWDPA} for a whole range of small approximation ratios from below. This is achieved by reducing the online problem of binary string guessing with known history (\textnormal{2-SGKH}) to the priority problem \textnormal{LWDPA}.

The binary string guessing problem was introduced by Böckenhauer et al.~\cite{bockenhauer2014string} with the aim of deriving such lower bounds for online algorithms. Importantly, in our case we do not reduce \textnormal{2-SGKH} to an online problem, but rather to a priority problem, which will make the construction slightly more complicated.

\begin{definition}[\textnormal{2-SGKH}]
The input of \textnormal{2-SGKH} is a parameter $n\in \N$, followed by binary values $d_1,\, d_2,\, ...\,,\, d_n \in \{0,1\}$ which are revealed one-by-one (in an online fashion). Before reading $d_i,\; i\in \{1,2,...\,,n\}$, an online algorithm solving \textnormal{2-SGKH} outputs $y_i\, \in \{0,1\}$. The gain of the solution produced by the algorithm is the number of $i$ in $\{1,\,2,\,...\,,\,n\}$ with $y_i=d_i$.
\end{definition}

In other words, an online algorithm solving \textnormal{2-SGKH} tries to guess $n$ binary values $d_1,\, d_2,\, ...\,,\, d_n$. Importantly, the algorithm knows the number $n$ of values to guess before making the first guess and is informed whether it guessed correctly after each guess. Böckenhauer et al.~\cite{bockenhauer2014string} proved the following theorem:

\begin{theorem}[{{Böckenhauer et al.~\cite[Corollary~1]{bockenhauer2014string}}}]\label{2sgkh}
Let $1/2 \leq \varepsilon <1$. An online algorithm for \textnormal{2-SGKH} guessing more than $\varepsilon{}n$ out of a total of $n$ bits correctly needs to read at least $(1-\mathcal{H}(\varepsilon))n$ advice bits, where $\mathcal{H}(x)=-x\log(x) - (1-x)\log(1-x)$ is the binary entropy function.
\end{theorem}

This theorem allows us to prove the following lower bound by reduction. A similar, more general construction was used by Boyar et al.~\cite{boyar2019advice} and applied to a wide range of optimization problems. 

\begin{theorem}\label{lowlwdpaadv}
Let $1/2 \leq \varepsilon < 1$. The advice complexity of any adaptive-priority algorithm for \textnormal{LWDPA} with an approximation ratio smaller than $3/(2+\varepsilon)$ is at least $(1-\mathcal{H}(\varepsilon))\,\lfloor l/3 \rfloor $.
\end{theorem}

\begin{proof}
Let \textsc{ALG} be an adaptive-priority algorithm for \textnormal{LWDPA} reading fewer than \((1-\mathcal{H}(\varepsilon))\, \lfloor l/3 \rfloor \) advice bits. Using \textsc{ALG} as a subroutine, we construct the online algorithm \textsc{GUESS} for \textnormal{2-SGKH}.

First, \textsc{GUESS} receives as its input the length $n\in \N$ of the binary string it has to guess. Next, it inputs the path $P=0,\,1,\,2,\,...\,,\,3n$ of length $l=3n$ to \textsc{ALG}, signifying that \textsc{ALG} will operate on an instance defined on $P$. Note that, strictly speaking, as discussed in the introduction, \textsc{GUESS} does not actually feed $P$ to \textsc{ALG}, but conveys it via the superscript of the requests we will define below.

\begin{figure}[t]
\centering
\resizebox{\textwidth}{!}{
\medskip
\begin{tikzpicture}

\node[] at (0,0) (node0) {};

\node[circle,fill=black, inner sep=0, minimum size=4] at (4,0) (node4) {};

\draw[dashed, line width=1] (node0)--(node4);

\foreach \n in {5,...,15}{
		\pgfmathtruncatemacro{\nminusone}{\n - 1}
       \node[circle,fill=black, inner sep=0, minimum size=4] at (\n,0) (node\n) {};
       \draw[line width=1] (node\nminusone)--(node\n);
    }

\node[] at (19,0) (node19) {};
\draw[dashed, line width=1] (node15)--(node19);

\node [circle, inner sep=0, minimum size=3] at (8,-0.7) (p11) {};
\node [circle, inner sep=0, minimum size=3] at (9, -0.7) (p12) {};
\draw [line width=1] (p11)--(p12);

\node[circle] at (7.5,-0.7) (label1) {$r_{i,1}$};

\node [circle, inner sep=0, minimum size=3] at (8,-1.1) (p21) {};
\node [circle, inner sep=0, minimum size=3] at (10,-1.1) (p22) {};
\draw [line width=1] (p21)--(p22);

\node[circle] at (7.5, -1.1) (label2) {$r_{i,2}$};

\node [circle, inner sep=0, minimum size=3] at (9, -1.5) (p31) {};
\node [circle, inner sep=0, minimum size=3] at (11, -1.5) (p32) {};
\draw [line width=1] (p31)--(p32);

\node[circle] at (8.5,-1.5) (label3) {$r_{i,3}$};

\node [circle, inner sep=0, minimum size=3] at (10, -1.9) (p41) {};
\node [circle, inner sep=0, minimum size=3] at (11, -1.9) (p42) {};
\draw [line width=1] (p41)--(p42);

\node[circle] at (9.5,-1.9) (label4) {$r_{i,4}$};

\draw [line width=1, decorate,
    decoration = {brace}] (5.05, 0.5) --  (7.95, 0.5);
\draw [line width=1, decorate,
    decoration = {brace}] (8.05, 0.5) --  (10.95, 0.5);
\draw [line width=1, decorate,
    decoration = {brace}] (11.05, 0.5) --  (13.95, 0.5);

\node[circle] at (6.5,0.9) (label4) {$p_{i-1}$};
\node[circle] at (9.5,0.9) (label4) {$p_i$};
\node[circle] at (12.5,0.9) (label4) {$p_{i+1}$};

\end{tikzpicture}
}
\caption{The four requests $r_{i,j}$, $1\leq j\leq 4$, defined on 
the subpath $p_i$. Note that these requests constitute two pairs of 
complementary paths in $p_i$.
\vspace{3.75ex}}\label{figlowlwdpaadv}
\end{figure}
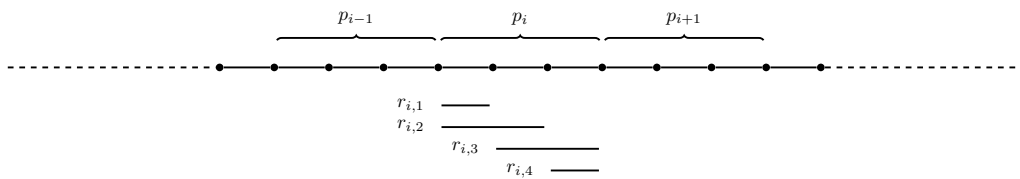

Divide the path $P$ into $n$ subpaths $p_1,\,p_2,\,...,\,p_n$ of length 3 each. Thus, $p_i$ is the path from vertex $3i-3$ to vertex $3i$. For each $1\leq i \leq n$, we define the following four requests, all of them subpaths of $p_i$ (see \cref{figlowlwdpaadv}):
\[r_{i,1}:=[3i-3,\,3i-2]^P, \; \; r_{i,2}:=[3i-3,\,3i-1]^P, \;\; r_{i,3}:=[3i-2,\,3i]^P, \;\; r_{i,4}:=[3i-1,\,3i]^P\]
Define $R_i= \{r_{i,j}\;|\;1\leq j\leq 4\}$ and, for each $r_{i,j}\in R_i$, denote by $r_{i,j}^c$ its complement in $p_i$, i.e., $r_{i,1}^c=r_{i,3}$, $r_{i,2}^c=r_{i,4}$, and vice versa. The idea is for \textsc{GUESS} to translate guessing each of the $n$ binary values into making \textsc{ALG} decide whether to accept or reject a certain request $r_{i,j}$ in each $p_i$. The difficulty lies in the fact that \textsc{GUESS} may not feed requests to \textsc{ALG} in an arbitrary way, but that it must respect the given adaptive-priority order. To this end, \textsc{GUESS} keeps track of two sets of requests, $U$ and $Q$, while processing its own input. Initially, $Q$ is empty and $U=\{r_{i,j}\;|\;1\leq i\leq n,\;1\leq j\leq 4\}$ contains all possible requests. \textsc{GUESS} then feeds the current highest-priority element in $U$ to \textsc{ALG} and guesses 1 if and only if \textsc{ALG} accepts that element. After having learned whether it should have guessed 1 or 0, \textsc{GUESS} can accordingly decide whether accepting or rejecting the presented element would have been optimal for \textsc{ALG}. It does this by either feeding the complement in $p_i$ of the presented path to \textsc{ALG}, in which case accepting (both requests) would have been optimal, or -- if rejecting would have been optimal -- by presenting the other two requests in $R_i$. However, \textsc{GUESS} may not feed these remaining requests to \textsc{ALG} instantly, but only after the adaptive-priority order permits so, which is why \textsc{GUESS} keeps track of the additional set $Q$ containing still-to-be-presented requests. Before each new guess, \textsc{GUESS} inputs high-priority requests remaining in $Q$ from previous guesses to \textsc{ALG}. More formally, denoting by $\max{(U\cup Q)}$ the element in $U\cup Q$ which has highest priority with respect to the current priority order, \textsc{GUESS} works as in Algorithm \ref{guessalgo}.

\begin{algorithm}

\caption{\textsc{GUESS}}\label{guessalgo}
\medskip

\textbf{Online Input: } The sequence $n,\,d_1,\, d_2,\, ...,\, d_n$, where $n\in \N$ and $d_i\in \{0,1\}$ for $i\in \{1,\,2,\,...\,,\,n\}$.

\medskip

\textbf{Advice: } The advice tape $\Phi$.

\medskip

\textbf{Online Output: } The sequence $y_1,\, y_2,\, ...,\, y_n$, where  $y_i\in \{0,1\}$ for $i\in \{1,\,2,\,...\,,\,n\}$.

\medskip

\newlength{\characterlength}
\settowidth{\characterlength}{(}

\textbf{Algorithm: }

\medskip

	\begin{algorithmic}
		\State Read the length $n\in \N$ of the binary string to guess.
		\State \hspace{-\characterlength}(Input the path $P=0,\,1,\,2,\,...\,,\,3n$ to \textsc{ALG}.) \Comment{Only informally.}
		\State $Q\gets \emptyset$
		\State $U\gets \{r_{i,j}\;|\;1\leq i\leq n,\;1\leq j\leq 4\}$
		\For {$1\leq k\leq n$}
			\While {$\max{(U\cup Q)} \in Q$}
				\State $m \gets \max{(U\cup Q)}$
				\State Feed $m$ to \textsc{ALG}
				\State $Q \gets Q\setminus \{m\}$
			\EndWhile
			\State $m_k\gets \max{(U\cup Q)}$ \Comment{$m_k\in U$.}
			\State Find $1 \leq i\leq n$ such that $m_k\in R_i$
       	\State $U \gets U \setminus R_i$
       	\State Feed $m_k$ to \textsc{ALG}, if \textsc{ALG} accepts, guess $y_k:=1$, else guess $y_k:=0$.
       	\State Read the true value $d_k$.
       	\If {$d_k=1$}
				\State $Q\gets Q\cup \{m_k^c\}$
			\EndIf
			\If {$d_k=0$}
				\State $Q\gets Q\cup (R_i\setminus\{m_k,\,m_k^c\})$
			\EndIf
		\EndFor
		\While {$Q \neq \emptyset$} \Comment{Post-processing}
			\State $m \gets \max Q$
			\State Feed $m$ to \textsc{ALG}
			\State $Q \gets Q\setminus \{m\}$
		\EndWhile
	\end{algorithmic}
\end{algorithm}

Since \textsc{GUESS} always feeds that request to \textsc{ALG} which currently has highest priority among all requests that could possibly still be presented, \textsc{GUESS} is well-defined. 

\textsc{GUESS} wrongly guesses the $k$'th bit $d_k$ if and only if \textsc{ALG} makes the wrong decision (with respect to the optimal solution) on whether to accept the current highest-priority element $m_k$. In this case, at most two out of three edges of $p_i$ can be accepted by \textsc{ALG}, whereas the optimal solution contains all three. This must happen for at least $(1-\varepsilon)n$ many $k$'s for some instance $\tilde{I}$ of \textnormal{2-SGKH}, $\abs{\tilde{I}}=n$, since otherwise, \textsc{GUESS} would consistently guess more than $\varepsilon n$ bits correctly, contradicting the fact that it uses fewer than \((1-\mathcal{H}(\varepsilon))\, \lfloor l/3 \rfloor=(1-\mathcal{H}(\varepsilon))n\) advice bits. Thus, denoting by $I$ the instance \textsc{GUESS} produces for \textsc{ALG} when processing $\tilde{I}$, we have
\begin{equation*} 
\frac{|\textsc{OPT}(I)|}{| \textsc{ALG}(I)|}\geq\frac{3n}{2(1-\varepsilon ) n + 3 \varepsilon n}=\frac{3}{2+\varepsilon}.
\end{equation*}
Note that this proves the theorem only for values of $l$ that are multiples of $3$. However, in order to prove the theorem for general $l$, we may simply adapt \textsc{GUESS} to use paths of length $l=3n+1$ or $l=3n+2$, defining all the necessary requests on the first $3n=3 \lfloor l/3 \rfloor$ edges.
\end{proof}

Letting $\varepsilon \to 1$ yields a lower bound of $\lfloor l/3 \rfloor$ on the number of advice bits necessary to achieve optimality, as compared to the $3\lceil{\frac{l}{4}}\rceil$ upper bound derived in \cref{uplwdpaadv}.

\section{DPA on Trees}\label{sec:CAT}

Returning to the original \textnormal{DPA} problem, we now consider the natural generalization of \textnormal{DPA} on paths, namely \textnormal{DPA} on the graph class $\mathcal{T}$ of trees (also referred to as the CAT problem).

As in the previous sections, we identify the request $[x,\,y]^T,\;T\in \mathcal{T}$, with the unique $x$-$y$-path in $T$. Since $T$ is always clear from the context, we will omit it as the superscript of $[x,\,y]$. As is common, we denote by $\deg(x)$ the degree of the vertex $x$ and by $d(x,\,y)$ the distance between the vertices $x$ and $y$, i.e., the length of the unique $x,y$-path in $T$.

\begin{theorem}\label{upcat}
There is a strictly 2-competitive fixed-priority algorithm without advice for \textnormal{DPA} on trees.
\end{theorem}

\begin{proof}
We again use the algorithm \textsc{GREEDY} that accepts every unblocked request. Here, \textsc{GREEDY} uses a priority order that orders paths primarily by non-increasing distance to the root of the tree $T$ on which the given instance is defined. More formally, choose, for each $T \in \mathcal{T}$, some arbitrary leaf $w$ as its root and define, for each path $p$ in $T$, its \emph{peak} to be the vertex $s_p:=\argmin_{v\in p}d(v,\,w)$, i.e., $p$'s vertex closest to the root.
\begin{claim}\label{upcatlem1}
The peak $s_p$ is unique.
\end{claim}

\begin{proof}
This follows quite easily from the fact that paths in trees are unique, i.e., that they do not contain cycles. Let $s_p,\,s_p'\in p$ be such that $d(s_p,\,w)=d(s_p',\,w)=\min _{v\in p}d(v,\,w)$. The $s_p$-$s_p'$-subpath of $p$ must be internally disjoint from the $s_p'$-$w$-path since if there was an internal intersection vertex $v$, the $v$-$w$-subpath of the $s_p'$-$w$-path would be a proper subpath and thus shorter. Thus, concatenating the $s_p$-$s_p'$-path and the $s_p'$-$w$-path yields again a path, which must be the original $s_p$-$w$-path since paths in trees are unique. Therefore, the length of the $s_p$-$w$-path equals the sum of the lengths of the $s_p$-$s_p'$-path and the $s_p'$-$w$-path, which shows that the length of the $s_p$-$s_p'$-path must be $0$.
\end{proof}

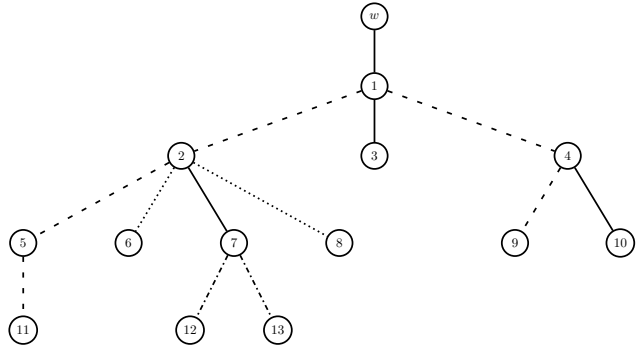
\begin{figure}[t]
\centering
\resizebox{.6\textwidth}{!}{
\begin{tikzpicture}[
  every node/.style = {shape=circle, solid, line width=1,
    draw, minimum size=2em , scale=0.75},
  level 1/.style={sibling distance=1em, level distance=4em},
  level 2/.style={sibling distance=11em, level distance=4em},
  level 3/.style={sibling distance=6em, level distance=5em},
  level 4/.style={sibling distance=5em, level distance=5em}
	]
	\node{$w$}
		child{
  			node[color=black] {1}
  			edge from parent [solid, line width=1]
    		child { 
     			node [color=black] {2}
     			edge from parent [loosely dashed, line width=1, color= black]
				child {
					node [color=black] {5} 
					edge from parent [loosely dashed, line width=1, color= black]
					child {
						node [color=black]{11} 
						edge from parent [loosely dashed, line width=1, color= black]}}
      			child {
      				node [color=black]{6}
    				edge from parent [dotted, line width=1, color= black]}
     		 	child {
      				node [color=black] {7}
     				edge from parent [solid, line width=1, color=black]
      				child{
     					node [color=black]{12}
      					edge from parent [dashdotted, line width=1, color= black]}
    				child{
      					node [color=black] {13}
      					edge from parent [dashdotted, line width=1, color= black]}}
      			child {
     				node [color=black]{8}
      				edge from parent [dotted, line width=1, color= black]}}
    		child {
     			node [color=black] {3}
    			edge from parent [solid, line width=1]}
     		child {
     			node [color=black] {4}
    			edge from parent [loosely dashed, line width=1, color= black]
				child {
					node [color=black] {9}
					edge from parent [loosely dashed, line width=1, color= black]}
      			child {
     				node [color=black]{10}
      				edge from parent [solid, line width=1, color=black]}}};
       
\end{tikzpicture}
}
\caption{Example of a tree and an instance for the proof of \cref{upcat}. The instance consists of the path $[12,\,13]$ with peak $s_{[12,\,13]}=7$, the path $[6,\,8]$ with peak $s_{[6,\,8]}=2$ and the path $[11,\,9]$ with peak $s_{[11,\,9]}=1$. With the priority order defined below, we have that $[12,\,13] \succ [6,\,8] \succ [11,\,9]$.}
\label{figupcat}
\end{figure}
We choose a fixed-priority order $\prec$ such that, for each $T \in \mathcal{T}$ and all requests $p$ and $p'$ defined on $T$,
\[d(s_p,\,w)>d(s_{p'},\,w)\implies p\succ p'.\]
See \cref{figupcat} for an example. With this priority order, we procede by induction on the size of an instance to show that \textsc{GREEDY} is strictly 2-competitive.

Since the claim is trivial for instances consisting of only one request, let $I$ be an instance of size at least two. As in the proof of \cref{thm:optdpa}, define $p:=\min_{\prec} \textsc{GREEDY}(I)$ and let $I'$ be the instance consisting of the same requests as $I$, except that $p$ and all requests with lower priority that intersect it are omitted. We again have that $|\textsc{GREEDY}(I)|=|\textsc{GREEDY}(I')|+1$. We now distinguish two cases:

\begin{description}
\item[\textit{Case~1.}] Consider first the case where the vertex $s_p$ is an end-vertex of the path $p$. Let $\{x,\,s_p\}$ be the edge of $p$ incident to $s_p$ and consider a path $q\in I \setminus I'$.
\begin{claim}\label{upcatlem2}
The path $q$ contains the edge $\{x,\,s_p\}$.
\end{claim}
\begin{proof}
By definition of $I'$, $q$ has lower priority than $p$ and intersects it on some edge. Thus, $p$ and $q$ share at least two common vertices. Let $v\neq s_p$ be such a vertex, as depicted in \cref{figupcatlem2}. As in \cref{upcatlem1}, the (non-trivial) $v$-$s_p$-subpath of $p$ must be internally disjoint from the $s_p$-$w$-path and the $v$-$s_q$-subpath of $q$ must be internally disjoint from the $s_q$-$w$-path -- by definition of the peaks $s_p$ and $s_q$. Thus, since paths in trees are unique, the concatenation of the $v$-$s_p$-path with the $s_p$-$w$-path coincides with the concatenation of the $v$-$s_q$-path with the $s_q$-$w$-path. Since $q$ has lower priority than $p$, the length of the $s_q$-$w$-path is at most the length of the $s_p$-$w$-path, which yields that the $v$-$s_p$-path, and thus in particular the edge $\{x,\,s_p\}$, is contained in the $v$-$s_q$-path and thus in $q$.
\end{proof}
The claim shows that at most one of the requests in $I\setminus I'$ can be part of $\textsc{OPT}(I)$ since if there were multiple, they would intersect on the edge $\{x,\,s_p\}$. Thus, $|\textsc{OPT}(I)|\leq|\textsc{OPT}(I')|+1$.
\item[\textit{Case~2.}] Consider now the case where $s_p$ is not an end-vertex of $p$. Let $\{x,\,s_p\}$ and $\{y,\,s_p\}$ be the two edges of $p$ incident to $s_p$. Split $p$ into two paths $p_x$ and $p_y$, both with end-vertex $s_p$, such that $\{x,\,s_p\}$ and $\{y,\,s_p\}$ are the end-edges of $p_x$ and $p_y$, respectively. A path $q\in I\setminus I'$ must intersect $p$ on some edge and thus either $p_x$ or $p_y$ or both. Furthermore, it trivially holds that $s_{p_x}=s_{p_y}=s_p$. The situation is thus the same as in the first case, and we get that $q$ must contain either the edge $\{x,\,s_p\}$ or the edge $\{y,\,s_p\}$ or both, depending on which of the two paths $q$ intersects. Thus, at most two requests in $I\setminus I'$ can be part of $\textsc{OPT}(I)$ since if there were three distinct such requests, two of them would intersect on $\{x,\,s_p\}$ or on $\{y,\,s_p\}$.
\end{description}
Hence, in both cases it holds that $|\textsc{OPT}(I)|\leq|\textsc{OPT}(I')|+2$. Applying the induction hypothesis to $I'$ yields the desired result:
\[\frac{\abs{\textsc{OPT}(I)}}{\abs{\textsc{GREEDY}(I)}}\leq \frac{\abs{\textsc{OPT}(I')}+2}{\abs{\textsc{GREEDY}(I')}+1}\leq \frac{2 \, \abs{\textsc{GREEDY}(I')}+2}{\abs{\textsc{GREEDY}(I')}+1}=2. \qedhere \]
\end{proof}

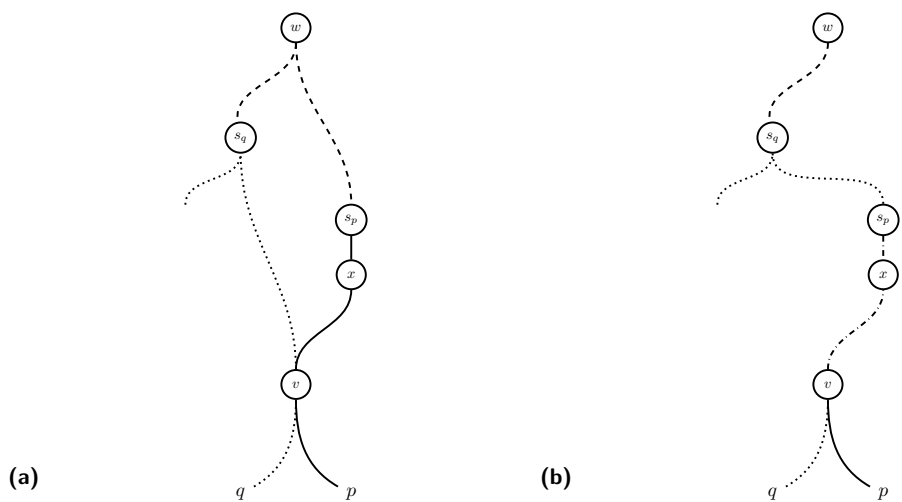
\begin{figure}[t]
\centering
\begin{subfigure}{.5\textwidth}
\centering

\resizebox{.4\textwidth}{!}{

\begin{tikzpicture}

	\node[shape=circle, solid, line width=1, draw, minimum size=2em, scale=0.75] at (0,4.5) (r) {$w$};
	\node[shape=circle, solid, line width=1, draw, minimum size=2em , scale=0.75] at (-1,2.5) (sp) {$s_q$};
	\draw[dashed, line width=1] (r) to[out=-90, in=100](sp);
	\node[shape=circle, solid, line width=1, draw, minimum size=2em , scale=0.75] at (1,1) (sq) {$s_p$};
	\draw[dashed, line width=1] (r) to[out=-90, in=90](sq);
	\node[shape=circle, solid, line width=1, draw, minimum size=2em , scale=0.75] at (1,0) (x) {$x$};
  	\draw[solid, line width=1] (sq)--(x);
	\node[shape=circle, solid, line width=1, draw, minimum size=2em , scale=0.75] at (0,-2) (v) {$v$};
  	\draw[solid, line width=1] (v) to[out=90, in=-90](x);
  	\draw[dotted, line width=1] (v) to[out=90, in=-90](sp);

	\node[shape=circle, solid, line width=1, minimum size=2em , scale=0.75] at (-2,1) (inv) {};
  	\draw[dotted, line width=1] (inv) to[out=90, in=-90](sp);

	\node[shape=circle, solid, line width=1, minimum size=2em , scale=0.75] at (-1,-4) (inv2) {};
  	\draw[dotted, line width=1] (inv2) to[out=30, in=-90] (v);
	\node[shape=circle, solid, line width=1, minimum size=2em , scale=0.75] at (1,-4) (inv3) {};
  	\draw[solid, line width=1] (inv3) to[out=150, in=-90] (v);

  	\node[shape=circle, solid, line width=1, minimum size=2em , scale=1] at (-1,-4) (inv4) {$q$};
  	\node[shape=circle, solid, line width=1, minimum size=2em , scale=1] at (1,-4) (inv4) {$p$};

\end{tikzpicture}
}
\vspace*{-5.5ex}
\caption{}
\label{figupcatlem2a}
\end{subfigure}%
\begin{subfigure}{.5\textwidth}
\centering

\resizebox{.4\textwidth}{!}{

\begin{tikzpicture}

	\node[shape=circle, solid, line width=1, draw, minimum size=2em , scale=0.75] at (0,4.5) (r) {$w$};
	\node[shape=circle, solid, line width=1, draw, minimum size=2em , scale=0.75] at (-1,2.5) (sp) {$s_q$};
	\draw[dashed, line width=1] (r) to[out=-90, in=100](sp);
	\node[shape=circle, solid, line width=1, draw, minimum size=2em , scale=0.75] at (1,1) (sq) {$s_p$};
	\draw[dotted, line width=1] (sp) to[out=-90, in=90](sq);
	\node[shape=circle, solid, line width=1, draw, minimum size=2em , scale=0.75] at (1,0) (x) {$x$};
  	\draw[dashdotted, line width=1] (sq)--(x);
	\node[shape=circle, solid, line width=1, draw, minimum size=2em , scale=0.75] at (0,-2) (v) {$v$};
  	\draw[dashdotted, line width=1] (v) to[out=90, in=-90](x);

	\node[shape=circle, solid, line width=1, minimum size=2em , scale=0.75] at (-2,1) (inv) {};
  	\draw[dotted, line width=1] (inv) to[out=90, in=-90](sp);

	\node[shape=circle, solid, line width=1, minimum size=2em , scale=0.75] at (-1,-4) (inv2) {};
  	\draw[dotted, line width=1] (inv2) to[out=30, in=-90] (v);
	\node[shape=circle, solid, line width=1, minimum size=2em , scale=0.75] at (1,-4) (inv3) {};
  	\draw[solid, line width=1] (inv3) to[out=150, in=-90] (v);
  	
  	\node[shape=circle, solid, line width=1, minimum size=2em , scale=1] at (-1,-4) (inv4) {$q$};
  	\node[shape=circle, solid, line width=1, minimum size=2em , scale=1] at (1,-4) (inv4) {$p$};

\end{tikzpicture}
}
\vspace*{-5.5ex}
\caption{}
\label{figupcatlem2b}
\end{subfigure}
\caption{Subfigure~(a) shows the a priori situation with regards to the paths $q$ (dotted) and $p$ (solid), intersecting on the vertex $v$. Since paths in trees are unique, the two $v$-$w$-paths in Subfigure~(a) must actually be the same path, and thus Subfigure~(b) applies. Since $q$ has lower priority than $p$, the vertex $s_q$ must lay on the $s_p$-$w$-path (as depicted in the figure) and not vice versa. Thus $q$ contains the edge $\{x,\,s_p\}$.}
\label{figupcatlem2}
\end{figure}

By slightly refining the priority order we defined above, we obtain the generalization of the result in \cref{thm:optdpa} mentioned earlier. This refinement consists in additionally requiring that $p \succ p'$ whenever $s_p=s_{p'}$ and $s_p$ is an end-vertex of $p$ but not of $p'$.

\begin{proposition}[Generalization of \cref{thm:optdpa}]\label{upcatcor}
There is an optimal fixed-priority algorithm without advice for \textnormal{DPA} on the graph class of trees with maximum degree at most $3$.
\end{proposition}

\begin{proof}
We follow the lines of the proof of \cref{upcat}. Define the request $p$ and the instance $I'$ accordingly. As seen above, if $p$'s peak $s_p$ is an end-vertex of $p$, we have $|\textsc{OPT}(I)|\leq|\textsc{OPT}(I')|+1$. Thus, consider again the case where $s_p$ is not an end-vertex of $p$, but now under the additional assumption that $\deg(s_p)\leq3$.

For paths $q,\,q'\in I\setminus I'$, we know that $q$ and $q'$ have lower priority than $p$ and therefore, $s_p$ is also not an end-vertex of $q$ and $q'$ if they have peak $s_p$, by the second condition on the priority order. Just as in the previous proof, both $q$ and $q'$ must contain at least one of the two edges $\{x,\,s_p\}$ and $\{y,\,s_p\}$ of $p$ incident to $s_p$. Since $\deg(s_p)\leq 3$, there is at most one additional edge $\{z,\,s_p\}$ incident to $s_p$, thus $q$ and $q'$, if they are distinct, must intersect on at least one of the edges $\{x,\,s_p\}$, $\{y,\,s_p\}$, $\{z,\,s_p\}$ since they both contain exactly two of them. Hence, at most one of the two paths can be part of $\textsc{OPT}(I)$.

We therefore get that $|\textsc{OPT}(I)|\leq|\textsc{OPT}(I')|+1$ in both the first and the second case, which yields
\[\frac{\abs{\textsc{OPT}(I)}}{\abs{\textsc{GREEDY}(I)}}\leq \frac{\abs{\textsc{OPT}(I')}+1}{\abs{\textsc{GREEDY}(I')}+1}\leq \frac{\abs{\textsc{GREEDY}(I')}+1}{\abs{\textsc{GREEDY}(I')}+1}=1. \qedhere \]
\end{proof}

Having established this optimality result for the subclass of $\mathcal{T}$ containing trees with maximum degree at most 3, it remains to prove a lower bound for the class of trees with maximum degree at least 4.

\begin{theorem}\label{lowcat}
The approximation ratio of any adaptive-priority algorithm without advice for \textnormal{DPA} on trees with maximum degree at least $4$ is at least $2$.
\end{theorem}

\begin{proof}
Let \textsc{ALG} be an adaptive-priority algorithm for \textnormal{DPA} on trees, $T \in \mathcal{T}$ a tree with maximum degree at least 4, and $v_0 \in V(T)$ with $\deg(v_0) \geq 4$. Let $v_1$, $v_2$, $v_3$, $v_4$ be distinct neighbors of $v_0$ (see \cref{figlowcat}) and $p=[p_1,\,p_2]:=\max_{\prec}\{[v_i,\,v_j]\;|\;1 \leq i < j \leq 4\}$ the path of length 2 with highest initial priority $\prec$ among these vertices.

If \textsc{ALG} rejects $p$, it has no approximation ratio at all because it has a gain of $0$ on the instance consisting only of $p$, while the optimal solution has a gain of $1$ on this instance. If \textsc{ALG} accepts $p$, consider the instance consisting of $p$, followed by the requests $[p_1,\,x]$ and $[p_2,\,y]$, where $x\neq y \in\{v_1,\,...\,,\,v_4\}\setminus \{p_1,\,p_2\}$. \textsc{ALG} cannot accept these requests because they intersect $p$, and thus it has a gain of 1 on this instance, while the optimal solution does not include $p$ but the other two requests.
\end{proof}

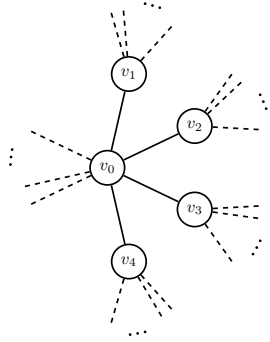
\begin{figure}[t]
\centering
\resizebox{.27\textwidth}{!}{
	\begin{tikzpicture}

	\node[circle, draw, line width=1] (center) {$v_0$};

	\node [circle, draw, line width=1] at ({3*360/14}:2) (n1) {$v_1$};
	\draw [line width=1](center)--(n1);

	\node [circle, draw, line width=1] at ({1*360/14}:2) (n2) {$v_2$};
	\draw [line width=1](center)--(n2);

	\node [circle, draw, line width=1] at ({(-1)*360/14}:2) (n3) {$v_3$};
	\draw [line width=1](center)--(n3);

	\node [circle, draw, line width=1] at ({(-3)*360/14}:2) (n4) {$v_4$};
	\draw [line width=1](center)--(n4);
	
	\node [circle, line width=1] at ({6*360/14}:2) (p01) {};
	\draw [line width=1, dashed](center)--(p01);
	\node [circle, line width=1] at ({7.5*360/14}:2) (p02) {};
	\draw [line width=1, dashed](center)--(p02);
	\node [circle, line width=1] at ({8*360/14}:2) (p03) {};
	\draw [line width=1, dashed](center)--(p03);

	\node [circle, fill=black, inner sep=0, draw] at ({6.575*360/14}:2) (q01) {};
	\node [circle, fill=black, inner sep=0, draw] at ({6.75*360/14}:2) (q02) {};
	\node [circle, fill=black, inner sep=0, draw] at ({6.925*360/14}:2) (q03) {};

	\node[circle, line width=1] at ({2.5*360/14}:3.5) (p11) {};
	\draw [line width=1, dashed](n1)--(p11);
	\node[circle, line width=1] at ({3.3*360/14}:3.5) (p12) {};
	\draw [line width=1, dashed](n1)--(p12);
	\node[circle, line width=1] at ({3.5*360/14}:3.5) (p13) {};
	\draw [line width=1, dashed](n1)--(p13);

	\node [circle, fill=black, inner sep=0, draw] at ({2.8*360/14}:3.5) (q11) {};
	\node [circle, fill=black, inner sep=0, draw] at ({2.9*360/14}:3.5) (q12) {};
	\node [circle, fill=black, inner sep=0, draw] at ({3*360/14}:3.5) (q13) {};

	\node[circle, line width=1] at ({0.5*360/14}:3.5)  (p21) {};
	\draw [line width=1, dashed](n2)--(p21);
	\node[circle, line width=1] at ({1.3*360/14}:3.5)  (p22) {};
	\draw [line width=1, dashed](n2)--(p22);
	\node[circle, line width=1] at ({1.5*360/14}:3.5)  (p23) {};
	\draw [line width=1, dashed](n2)--(p23);
	
	\node [circle, fill=black, inner sep=0, draw] at ({0.8*360/14}:3.5) (q21) {};
	\node [circle, fill=black, inner sep=0, draw] at ({0.9*360/14}:3.5) (q22) {};
	\node [circle, fill=black, inner sep=0, draw] at ({1*360/14}:3.5) (q23) {};	
	\node[circle, line width=1] at ({(-1.5)*360/14}:3.5)  (p31) {};
	\draw [line width=1, dashed](n3)--(p31);
	\node[circle, line width=1] at ({(-0.7)*360/14}:3.5)  (p32) {};
	\draw [line width=1, dashed](n3)--(p32);
	\node[circle, line width=1] at ({(-0.5)*360/14}:3.5)  (p33) {};
	\draw [line width=1, dashed](n3)--(p33);

	\node [circle, fill=black, inner sep=0, draw] at ({(-1.2)*360/14}:3.5) (q31) {};
	\node [circle, fill=black, inner sep=0, draw] at ({(-1.1)*360/14}:3.5) (q32) {};
	\node [circle, fill=black, inner sep=0, draw] at ({(-1)*360/14}:3.5) (q33) {};

	\node[circle, line width=1] at ({(-3.5)*360/14}:3.5)  (p41) {};
	\draw [line width=1, dashed](n4)--(p41);
	\node[circle, line width=1] at ({(-2.7))*360/14}:3.5)  (p42) {};
	\draw [line width=1, dashed](n4)--(p42);
	\node[circle, line width=1] at ({(-2.5))*360/14}:3.5)  (p43) {};
	\draw [line width=1, dashed](n4)--(p43);
	
	\node [circle, fill=black, inner sep=0, draw] at ({(-3.2)*360/14}:3.5) (q41) {};
	\node [circle, fill=black, inner sep=0, draw] at ({(-3.1)*360/14}:3.5) (q42) {};
	\node [circle, fill=black, inner sep=0, draw] at ({(-3)*360/14}:3.5) (q43) {};

	\end{tikzpicture}
}
\caption{The construction used in the proof of \cref{lowcat} for some tree $T$ with maximum degree at least $4$. Say, for example, that the request $[v_1,\,v_2]$ has highest initial priority among all paths of length $2$ in the subtree $T(\{v_0,\,v_1,\,...\,,\,v_4\}) \cong S_4= K_{1,4}$. Consider the instance consisting of $[v_1,\,v_2]$, followed by $[v_1,\,v_3]$ and by $[v_2,\,v_4]$. The algorithm must accept the first request and thus reject the other two, while the optimal solution rejects the first request but accepts the other two.}
\label{figlowcat}
\end{figure}

With this tight lower bound on the approximation ratio achievable without advice, we now turn to the advice setting. We begin with an optimal fixed-priority algorithm reading advice. As in the proof of \cref{uplwdpaadv}, the advice string is going to encode some optimal solution for the instance at hand, which the algorithm is going to decode. This fixed optimal solution will in fact again be the \emph{greediest} optimal solution, as defined in \cref{greedopt}.

At first, our result might seem surprising in that the number of advice bits the algorithm requires does not directly depend on the size of the tree on which the given instance is defined. However, this makes intuitive sense, considering that the \textnormal{DPA} problem on paths -- or, more generally, on trees with maximum degree at most 3 -- is optimally solvable without advice. We can thus expect the amount of advice necessary to achieve optimality on a certain tree to be small if it contains large subtrees without vertices of degree greater than $3$.

\begin{theorem}\label{upcatadv}
There is an optimal fixed-priority algorithm for \textnormal{DPA} on trees with advice complexity $(\theta_1(T)-\theta_3(T)-2) \left\lceil \log \frac{\Delta (T)}{2}\right\rceil $, where $\theta_1(T)$ denotes the number of leaves in the tree $T$, $\theta_3(T)$ the number of vertices of degree $3$ and $\Delta (T)$ the maximum degree in $T$.
\end{theorem}

\begin{proof}
We are again extending the partial order on $\mathcal{U}$ that we initially defined in the proof of \cref{upcat} and already extended once for \cref{upcatcor}. We pose the additional condition that all requests in the given instance that share the same peak are presented consecutively to our algorithm. To achieve this, we choose some arbitrary total order $<_T$ on the set $V(T)$ for each $T \in \mathcal{T}$ and define, for each $T \in \mathcal{T}$ and all requests $p$, $p'$ defined on $T$,
\[d(s_p,\,w)>d(s_{p'},\,w) \implies p\succ p',\]
\[d(s_p,\,w)=d(s_{p'},\,w)\;\textnormal{and}\;s_p >_{T} s_{p'} \implies p\succ p',\]
\[s_p=s_{p'}\;\textnormal{and}\;(\textnormal{$s_p$ is an end-vertex of $p$ but not of $p'$}) \implies p\succ p',\]
and extend this partial order to some total order on $\mathcal{U}$. With this priority order, the algorithm we define below will again be a quasi-greedy algorithm, as in the proof of \cref{uplwdpaadv}.

Let $\textsc{OPT}_{\textnormal{Gr}}(I)$ again be the greediest optimal solution for the instance $I$ with respect to $\prec$. In \textnormal{DPA}, we have a simple characterization of $\textsc{OPT}_{\textnormal{Gr}}(I)$. Namely, for all $p \succ p'\in I$ with $p' \in \textsc{OPT}_{\textnormal{Gr}}(I)\;\textnormal{and}\;p\notin \textsc{OPT}_{\textnormal{Gr}}(I)$, it holds that the set $(\textsc{OPT}_{\textnormal{Gr}}(I)\cup \{p\})\setminus \{p'\}$ is not a valid solution for $I$, i.e., it contains intersecting paths. This is because if it were a valid solution, it would in fact be an optimal solution, since it has the same size as $\textsc{OPT}_{\textnormal{Gr}}(I)$. This, however, means that $p$ should have been included in $\textsc{OPT}_{\textnormal{Gr}}(I)$ by \cref{greedopt}.

\begin{claim}\label{claim:peakingredopt}
Let $p\in I$ be a path such that its peak $s_p$ is an end-vertex of $p$ and such that $p$ does not intersect a path in $\textsc{OPT}_{\textnormal{Gr}}(I)$ with higher priority. It holds that $p\in \textsc{OPT}_{\textnormal{Gr}}(I)$.
\end{claim}

\begin{proof}
Assume the contrary, i.e., assume that $s_p$ is an end-vertex of $p$, that $p$ does not intersect a path in $\textsc{OPT}_{\textnormal{Gr}}(I)$ with higher priority, and that $p\notin \textsc{OPT}_{\textnormal{Gr}}(I)$. There must be a later request $p'\prec p$ contained in $\textsc{OPT}_{\textnormal{Gr}}(I)$ that intersects $p$ because otherwise, we could add $p$ to $\textsc{OPT}_{\textnormal{Gr}}(I)$ without creating an intersection, contradicting the fact that $\textsc{OPT}_{\textnormal{Gr}}(I)$ is optimal. As discussed in the proof of \cref{upcat}, $p'$ is unique because it contains the edge of $p$ incident to $s_p$. However, this means that $(\textsc{OPT}_{\textnormal{Gr}}(I)\cup \{p\})\setminus \{p'\}$ is also a valid solution for $I$, which is a contradiction.
\end{proof}

\begin{claim}\label{claim:degingredopt}
Let $p\in I$ be a path such that $\deg(s_p) \leq 3$ and such that $p$ does not intersect a path in $\textsc{OPT}_{\textnormal{Gr}}(I)$ with higher priority. It holds that $p\in \textsc{OPT}_{\textnormal{Gr}}(I)$.
\end{claim}

\begin{proof}
If we assume the contrary, there must again be some request with lower priority than $p$ that is contained in $\textsc{OPT}_{\textnormal{Gr}}(I)$, by optimality of $\textsc{OPT}_{\textnormal{Gr}}(I)$. As discussed in the proof of \cref{upcatcor}, $\deg(s_p) \leq 3$ implies that this later request is again unique, which just as above leads to a contradiction with the greediness of $\textsc{OPT}_{\textnormal{Gr}}(I)$.
\end{proof}

Thus, our algorithm can operate greedily on requests whose peaks coincide with one of their end points or have degree at most $3$. More precisely, if our algorithm must decide whether to accept or reject a request $p$ and we assume that its partial solution on the requests preceding $p$ agrees with $\textsc{OPT}_{\textnormal{Gr}}(I)$, accepting $p$ if it is not yet blocked and rejecting it otherwise extends this partial solution in such a way that it still agrees with $\textsc{OPT}_{\textnormal{Gr}}(I)$, given that $s_p$ is an end-vertex of $p$ or that $\deg(s_p)\leq 3$.

We now use the following strategy: We partition the instance $I$ into disjoint phases, such that phase $I_v$ contains all paths with peak $v$. Due to our condition on the priority order that all requests with peak $v$ are presented consecutively, the algorithm knows that the previous phase is finished and the new phase $I_v$ starts as soon as it first receives a request with the new peak $v$.

\begin{figure}[t]
\centering
\resizebox{.55\textwidth}{!}{
\begin{tikzpicture}

	\node[shape=circle, solid, line width=1, minimum size=2em, scale=0.75, draw] at (0,5.5) (w) {$w$};
	\node[shape=circle, solid, line width=1, minimum size=2em, scale=0.75, draw] at (0,4.25) (n1) {$1$};
	\draw[dashed, line width=2] (w)--(n1);
	\node[shape=circle, solid, line width=1, minimum size=2em, scale=0.75, draw] at (-1.25,4.25) (n2) {$2$};
	\draw[solid, line width=2, color=lipicsYellow] (n1)--(n2);
	\node[shape=circle, solid, line width=1, minimum size=2em, scale=0.75, draw] at (0,3) (n3) {$3$};
	\draw[solid, line width=2, color=lipicsYellow] (n1)--(n3);

	\node[shape=circle, solid, line width=1, minimum size=2em, scale=0.75, draw] at (-4.5,0) (n11) {$4$};
	\draw[solid, line width=2, color=LimeGreen] (n3) to[out=180, in=90] node[left, minimum size=4em, color=black]{$e_3^1$}(n11);
	\node[shape=circle, solid, line width=1, minimum size=2em, scale=0.75, draw] at (-3,0) (n12) {$5$};
	\draw[solid, line width=2, color=RedOrange] (n3) to[out=210, in=90] node[left, minimum size=3em, color=black]{$e_3^2$}(n12);
	\node[shape=circle, solid, line width=1, minimum size=2em, scale=0.75, draw] at (-1.5,0) (n13) {$6$};
	\draw[dashed, line width=2] (n3) to[out=240, in=90] node[left, minimum size=3em]{$e_3^3$}(n13);
	\node[shape=circle, solid, line width=1, minimum size=2em, scale=0.75, draw] at (0,0) (n14) {$7$};
	\draw[solid, line width=2, color=RedOrange] (n3) to[out=270, in=90] node[left, minimum size=3em, color=black]{$e_3^4$}(n14);
	\node[shape=circle, solid, line width=1, minimum size=2em, scale=0.75, draw] at (1.5,0) (n15) {$8$};
	\draw[solid, line width=2, color=LimeGreen] (n3) to[out=300, in=90] node[left, minimum size=3em, color=black]{$e_3^5$}(n15);
	\node[shape=circle, solid, line width=1, minimum size=2em, scale=0.75, draw] at (3,0) (n16) {$9$};
	\draw[solid, line width=2, color=NavyBlue] (n3) to[out=330, in=90] (n16);
	\node[shape=circle, solid, line width=1, minimum size=2em, scale=0.75, draw] at (4.5,0) (n17) {$10$};
	\draw[solid, line width=2, color=lipicsYellow] (n3) to[out=360, in=90] node[left, minimum size=5em, color=black]{$e_3^6$}(n17);
	
	\node[shape=circle, solid, line width=1, minimum size=2em, scale=0.75, draw] at (-4.5,-1.25) (n111) {$11$};
	\draw[solid, line width=2, color=LimeGreen] (n11) to[out=270, in=90](n111);

	\node[shape=circle, solid, line width=1, minimum size=2em, scale=0.75, draw] at (-2,-1.25) (n131) {$12$};
	\draw[solid, line width=2, color=Thistle] (n13)--(n131);
	\node[shape=circle, solid, line width=1, minimum size=2em, scale=0.75, draw] at (-1,-1.25) (n132) {$13$};
	\draw[solid, line width=2, color=Thistle] (n13)--(n132);

	\node[shape=circle, solid, line width=1, minimum size=2em, scale=0.75, draw] at (3,-1.25) (n161) {$14$};
	\draw[solid, line width=2, color=NavyBlue] (n16)--(n161);

	\node[shape=circle, solid, line width=1, minimum size=2em, scale=0.75, draw] at (4.5,-1.25) (n171) {$15$};
	\draw[solid, line width=2, color=lipicsYellow] (n17)--(n171);
    
\end{tikzpicture}
}
\caption{Example of a labeling in the proof of \cref{upcatadv}. Let $I$ be an instance defined on the tree above with $\textsc{OPT}_{\textnormal{Gr}}(I)$ consisting of the paths colored uniformly in the figure, i.e., the paths $[11,\,8]$, $[5,\,7]$, $[12,\,13]$, $[3,\,14]$ and $[2,\,15]$. Consider our algorithm in the phase $I_3$, containing all requests in $I$ with peak $3$. By definition of the priority order, the path $[3,\,14]$ with end-vertex $3$ has highest priority among all the requests in $I_3$ and is greedily accepted by the algorithm without requiring advice. The six remaining, i.e., still unblocked edges between the vertex $3$ and its children are denoted by $e_3^1,\,e_3^2,\,,,,\,,\,e_3^6$ and labeled. First, the edges $e_3^3$ and $e_3^6$ receive the label $0$ because they are not part of any optimal path ($e_3^3$) or part of an optimal path belonging to a later phase ($e_3^6$). Then, the edges $e_3^1$ and $e_3^5$ receive the same label, belonging to the optimal path $[11,\,8]$, and $e_3^2$, $e_3^4$ receive the same label, belonging to $[5,\,7]$. Thus, the sequence $1,\,2,\,0,\,2,\,1,\,(0)$ encodes the optimal solution for the phase $I_3$. The last label ($0$) does not need to be conveyed since the algorithm can infer it from the other labels.}
\label{figupcatadv}
\end{figure}
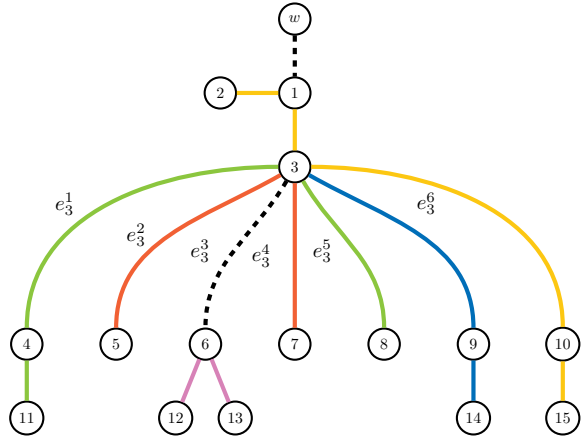

If $v$ has degree at most $3$, the algorithm operates greedily during the entire phase $I_v$ without reading advice. If $v$ has degree greater than $3$, the algorithm first greedily accepts all unblocked paths in $I_v$ with end point $v$, again requiring no advice. For the remaining requests in $I_v$, the algorithm needs advice:

Similarly to the advice strategy used by Böckenhauer et al.~\cite{bockenhauer2019call}, the basic idea is that the advice string communicates to the algorithm which pairs of edges incident to $v$ belong to the same optimal request in the remainder of $I_v$. More precisely, we first fix an ordering of the remaining edges $e_v^1,\,e_v^2,\,...,e_v^{n_v}$  between $v$ and its children. Remaining edges are those edges that are not yet blocked by accepted requests with end point $v$. Note that $n_v\leq \deg(v)-1$ since we chose a leaf as the root in the proof of \cref{upcat} and thus all vertices with degree greater than $3$ have a parent. Now, the edges $e_v^1,\,e_v^2,\,...,e_v^{n_v}$ each receive a label to denote which of them are part of the same path in $\textsc{OPT}_{\textnormal{Gr}}(I)$.

This is done by first labeling edges that are not part of any optimal path belonging to $I_v$ with $0$. Note that there are two kinds of such edges: On the one hand, there are those that are not part of any optimal path at all, and on the other, there possibly is one edge that is part of an optimal path belonging to a later phase.

There now remain $n_v'\leq n_v$ edges, constituting $n_v'/2$ pairs of edges belonging to the same optimal path in $I_v$. These edges receive labels in $\{1,\,2,\,...\,,\,n_v'/2\}$, such that two edges receive the same label if and only if they are contained in the same optimal path. See \cref{figupcatadv} for an example of a labeling.

Having the labels for phase $I_v$, the algorithm checks for each unblocked path $p \in I_v$ whether its edges incident to $v$ are labeled identically with a strictly positive label and accepts $p$ only if they are. If $p$ is blocked, the algorithm naturally rejects it.

\begin{claim}\label{claim:advworks}
If the algorithm's partial solution on the requests preceding $p$ agrees with $\textsc{OPT}_{\textnormal{Gr}}(I)$, this will still be true after accepting or rejecting $p$ according to this rule.
\end{claim}
\begin{proof}
If the algorithm rejects $p$, the claim is clear. If the algorithm accepts $p$, the labels must match and there must thus be some request $p' \preceq p,\;p' \in \textsc{OPT}_{\textnormal{Gr}}(I)$ containing the same two edges incident to $v$ as $p$, by definition of the labeling.

We know that $p$ does not intersect a path in $\textsc{OPT}_{\textnormal{Gr}}(I)$ with higher priority since the two solutions agree on those requests and the algorithm would thus have accepted $p$ even though it was blocked. Furthermore, $p$ also does not intersect a path $q \in \textsc{OPT}_{\textnormal{Gr}}(I)$ distinct from $p'$ with $q \prec p$ since by Case~2 in \cref{upcatlem2}, $q$ would contain at least one of the two edges of $p$ incident to $v$ and would thus intersect $p'$.
\newpage
Thus, $p'$ is the only request in $\textsc{OPT}_{\textnormal{Gr}}(I)$ that intersects $p$. This means that $(\textsc{OPT}_{\textnormal{Gr}}(I)\cup \{p\})\setminus \{p'\}$ is also a valid solution for $I$, which is a contradiction if $p' \succ p$. We therefore have that $p'=p$ which yields the claim.
\end{proof}

Inductively, we thus know that the algorithm perfectly recreates $\textsc{OPT}_{\textnormal{Gr}}(I)$ if it has the labels. It remains to bound the length of the advice string necessary to convey them.

For phases $I_v$ with $\deg(v)\leq3$, no advice is necessary. For the other phases, note that not all the labels $e_v^1,\,e_v^2,\,...,e_v^{n_v}$ need to be conveyed: The last label $e_v^{n_v}$ can be inferred by the algorithm by checking whether all the strictly positive labels appear twice among $e_v^1,\,e_v^2,\,...\,,e_v^{n_v-1}$. If they do, the last label must be 0, else, it must be the missing positive label. Thus, for each phase $I_v$ with $\deg(v)\geq4$, the advice needs to encode $n_v-1\leq \deg(v)-2$ numbers between $0$ and $\frac{n_v'}{2}$. Because $n_v'$ is even, we have $\frac{n_v'}{2}\leq \lfloor\frac{n_v}{2}\rfloor\leq \lfloor\frac{\deg(v)-1}{2}\rfloor$. Thus, the number of advice bits sufficient to convey the labels for phase $I_v$ is
\begin{align*}
(\deg(v)-2)\left\lceil\log\left(\left\lfloor \frac{\deg(v)-1}{2}\right\rfloor+1\right)\right\rceil&=(\deg(v)-2)\left\lceil\log \left\lceil \frac{\deg(v)}{2}\right\rceil \right\rceil\\
&=(\deg(v)-2) \left\lceil\log\frac{\deg(v)}{2} \right\rceil.
\end{align*}
Since this number only depends on the degree of the peak, the algorithm knows how many bits to read from the advice string in each phase and thus the strings for all the phases can be concatenated, resulting in an advice string of length
\[\sum_{\substack{v\\\deg(v)\geq 4}}{(\deg(v)-2)\left\lceil\log\frac{\deg(v)}{2}\right\rceil}.\]
We denote by $\Delta=\Delta (T)$ the maximum degree in the tree $T$ and by $\theta_1$, $\theta_2$ and $\theta_3$ the number of vertices of degree $1$, $2$ and $3$, respectively. Using that a tree with $m$ edges contains $m+1$ vertices and that $\sum_{v}{\deg(v)}=2m$ yields

\settowidth{\characterlength}{=(}

\begin{align*}
&\hspace{\characterlength} \sum_{\substack{v\\\deg(v)\geq 4}}{(\deg(v)-2) \left\lceil \log\frac{\deg(v)}{2}\right\rceil}\\
&\leq \sum_{\substack{v\\\deg(v)\geq 4}}{(\deg(v)-2)} \left\lceil \log \frac{\Delta}{2} \right\rceil\\
&=\left(\sum_{v}{(\deg(v)-2)}-\sum_{\substack{v\\\deg(v)\leq 3}}{(\deg(v)-2)}\right) \left\lceil \log \frac{\Delta}{2} \right\rceil\\
&=(2m-2(m+1)-(\theta_1 \cdot(-1)+\theta_2 \cdot 0 +\theta_3 \cdot 1)) \left\lceil \log \frac{\Delta}{2} \right\rceil\\
&=(\theta_1-\theta_3-2) \left\lceil \log \frac{\Delta}{2} \right\rceil.\qedhere
\end{align*}
\end{proof}

Lastly, we again use a reduction of the online binary string guessing problem to derive lower bounds on the advice complexity of priority algorithms for \textnormal{DPA} on trees with small approximation ratios. As in the proof of \cref{lowlwdpaadv}, the core idea is to translate guessing a bit into guessing whether a certain request should be accepted or not. We will use a similar technique as in the proof of \cref{lowcat}, making this request the length-two path with highest priority in a subtree isomorphic to the star tree $S_4$, i.e., the complete bipartite graph $K_{1,4}$. We require the following lemma.

\begin{lemma}\label{lowcatlem}
Let $T$ be a tree. $T$ contains at least $\left \lceil \frac{1}{2}\sum_{v \in V(T)}{\left\lfloor \frac{\deg(v)}{4}\right\rfloor}\right \rceil$ pairwise edge-disjoint subtrees isomorphic to $S_4$.
\end{lemma}
\begin{proof}
We use induction on $s(T):=\sum_{v \in V(T)}{\left\lfloor \frac{\deg(v)}{4}\right\rfloor}$. For $s(T)=0$, the claim is trivial. For $s(T) \geq 1$, let $u \in V(T)$ be a vertex of degree at least $4$ such that $u$ has at most one neighbor whose degree is also greater than or equal to $4$. If no such $u$ exists, let $H$ be a connected component of the subgraph of $T$ induced by the set of vertices of degree at least $4$. Because the minimal degree of $H$ is at least $2$, we get a contradiction since a connected subgraph of a tree is also a tree and nonempty trees have minimal degree $1$.

Let $G'$ be the forest that results from deleting $u$ from $G$. $G'$ is the disjoint union of $k:=\deg(u)$ trees $T_1,\,T_2,\,...\,,\,T_k$. Since deleting $u$ reduces the degree of at most one other vertex with degree at least $4$ by at most $1$, we have that
\[s(G')=\sum_{i=1}^{k}\sum_{v \in V(T_i)}{\left\lfloor \frac{\deg_{T_i}(v)}{4}\right\rfloor} \geq s(T)-\left\lfloor\frac{\deg(u)}{4}\right\rfloor-1.\]
On the other hand, $s(T_i) \leq s(G') \leq s(T)-\left\lfloor\deg(u)/4\right\rfloor \leq s(T)-1$ for all $1 \leq i \leq k$. Thus, by the induction hypothesis, each $T_i$ contains at least $s(T_i)/2$ pairwise edge-disjoint copies of $S_4$. In addition, $T$ contains $\left\lfloor\deg(u)/4\right\rfloor$ edge-disjoint copies, each consisting of $u$ and $4$ of its neighbors. This yields a total of
\begin{align*}
\left\lfloor \frac{\deg(u)}{4}\right\rfloor + \sum_{i=1}^{k} \frac{s(T_i)}{2} &= \left\lfloor \frac{\deg(u)}{4}\right\rfloor + \sum_{i=1}^{k} \, \frac{1}{2} \sum_{v \in V(T_i)} \left\lfloor \frac{\deg_{T_i}(v)}{4}\right\rfloor \\
&\geq \left\lfloor \frac {\deg(u)}{4}\right\rfloor + \frac{1}{2} \left(s(T)-\left\lfloor\frac{\deg(u)}{4}\right\rfloor-1 \right)\\
&=\frac{1}{2} \, s(T) + \frac{1}{2} \left( \left\lfloor \frac{\deg(u)}{4}\right\rfloor -1\right)\\
&\geq \frac{1}{2} \, s(T)
\end{align*}
pairwise edge-disjoint copies of $S_4$ in $T$.
\end{proof}

\begin{figure}[t]
\centering
\resizebox{.65\textwidth}{!}{
\begin{tikzpicture}

\node[circle, inner sep=0pt, minimum size=4pt] at (-2,0) (nodem2) {};

\node[circle, inner sep=0pt, minimum size=4pt] at (-1.5,0) (nodem15) {}; 
\draw[dashed, line width=1pt] (nodem2)--(nodem15);

\node[circle, fill=black, inner sep=0pt, minimum size=4pt] at (-1,0) (nodem1) {}; 
\draw[solid, line width=1pt] (nodem15)--(nodem1);

\node[circle, fill=black, inner sep=0pt, minimum size=4pt] at (0,0) (node0) {}; 
\draw[solid, line width=1pt] (nodem1)--(node0);

\node[circle, fill=black, inner sep=0pt, minimum size=4pt] at (1,0) (node1) {}; 
\draw[solid, line width=1pt] (node0)--(node1);

\node[circle, inner sep=0pt, minimum size=4pt] at (1.5,0) (node15) {}; 
\draw[solid, line width=1pt] (node1)--(node15);

\node[circle, inner sep=0pt, minimum size=4pt] at (2,0) (node2) {}; 
\draw[dashed, line width=1pt] (node15)--(node2);

\node[circle, fill=black, inner sep=0pt, minimum size=4pt] at (-1,1) (nodeum1) {}; 
\draw[solid, line width=1pt] (nodem1)--(nodeum1);
\node[circle, fill=black, inner sep=0pt, minimum size=4pt] at (-1,-1) (nodelm1) {}; 
\draw[solid, line width=1pt] (nodem1)--(nodelm1);

\node[circle, fill=black, inner sep=0pt, minimum size=4pt] at (0,1) (nodeu0) {}; 
\draw[solid, line width=1pt] (node0)--(nodeu0);
\node[circle, fill=black, inner sep=0pt, minimum size=4pt] at (0,-1) (nodel0) {}; 
\draw[solid, line width=1pt] (node0)--(nodel0);

\node[circle, fill=black, inner sep=0pt, minimum size=4pt] at (1,1) (nodeu1) {}; 
\draw[solid, line width=1pt] (node1)--(nodeu1);
\node[circle, fill=black, inner sep=0pt, minimum size=4pt] at (1,-1) (nodel1) {}; 
\draw[solid, line width=1pt] (node1)--(nodel1);

\node[circle, fill=black, inner sep=0pt, minimum size=4pt] at (-6,0) (nodem6) {};
\node[circle, fill=black, inner sep=0pt, minimum size=4pt] at (-5,0) (nodem5) {};
\draw[solid, line width=1pt] (nodem6)--(nodem5);
\node[circle, fill=black, inner sep=0pt, minimum size=4pt] at (-4,0) (nodem4) {}; 
\draw[solid, line width=1pt] (nodem5)--(nodem4);
\node[circle, inner sep=0pt, minimum size=4pt] at (-3.5,0) (nodem35) {}; 
\draw[solid, line width=1pt] (nodem4)--(nodem35);
\node[circle, inner sep=0pt, minimum size=4pt] at (-3,0) (nodem3) {};
\draw[dashed, line width=1pt] (nodem35)--(nodem3);

\node[circle, fill=black, inner sep=0pt, minimum size=4pt] at (-5,1) (nodeum5) {}; 
\draw[solid, line width=1pt] (nodem5)--(nodeum5);
\node[circle, fill=black, inner sep=0pt, minimum size=4pt] at (-5,-1) (nodelm5) {}; 
\draw[solid, line width=1pt] (nodem5)--(nodelm5);

\node[circle, fill=black, inner sep=0pt, minimum size=4pt] at (-4,1) (nodeum4) {}; 
\draw[solid, line width=1pt] (nodem4)--(nodeum4);
\node[circle, fill=black, inner sep=0pt, minimum size=4pt] at (-4,-1) (nodelm4) {}; 
\draw[solid, line width=1pt] (nodem4)--(nodelm4);

\node[circle, fill=black, inner sep=0pt, minimum size=1pt] at (-2.65,0) (p1) {};
\node[circle, fill=black, inner sep=0pt, minimum size=1pt] at (-2.5,0) (p2) {};
\node[circle, fill=black, inner sep=0pt, minimum size=1pt] at (-2.35,0) (p3) {}; 

\node[circle, fill=black, inner sep=0pt, minimum size=4pt] at (6,0) (node6) {};
\node[circle, fill=black, inner sep=0pt, minimum size=4pt] at (5,0) (node5) {};
\draw[solid, line width=1pt] (node6)--(node5);
\node[circle, fill=black, inner sep=0pt, minimum size=4pt] at (4,0) (node4) {}; 
\draw[solid, line width=1pt] (node5)--(node4);
\node[circle, inner sep=0pt, minimum size=4pt] at (3.5,0) (node35) {}; 
\draw[solid, line width=1pt] (node4)--(node35);
\node[circle, inner sep=0pt, minimum size=4pt] at (3,0) (node3) {}; 
\draw[dashed, line width=1pt] (node35)--(node3);
\node[circle, fill=black, inner sep=0pt, minimum size=4pt] at (4,1) (nodeu4) {};

\node[circle, fill=black, inner sep=0pt, minimum size=4pt] at (5,1) (nodeu5) {}; 
\draw[solid, line width=1pt] (node5)--(nodeu5);
\node[circle, fill=black, inner sep=0pt, minimum size=4pt] at (5,-1) (nodel5) {}; 
\draw[solid, line width=1pt] (node5)--(nodel5);

\draw[solid, line width=1pt] (node4)--(nodeu4);
\node[circle, fill=black, inner sep=0pt, minimum size=4pt] at (4,-1) (nodel4) {}; 
\draw[solid, line width=1pt] (node4)--(nodel4);

\node[circle, fill=black, inner sep=0pt, minimum size=1pt] at (2.65,0) (q1) {};
\node[circle, fill=black, inner sep=0pt, minimum size=1pt] at (2.5,0) (q2) {};
\node[circle, fill=black, inner sep=0pt, minimum size=1pt] at (2.35,0) (q3) {}; 

\draw [line width=1pt, decorate,
    decoration = {brace}] (6, -1.5)--(-6, -1.5);
\node[circle] at (0,-2) (label1) {$2n+1$};

\end{tikzpicture}
}
  \caption{Example of a tree $T$ with $ \left \lceil \frac{1}{2}\sum_{v \in V(T)} \left\lfloor \frac{\deg(v)}{4}\right\rfloor \right \rceil= n$. It consists of a path of length $2n+1$ with two leaves attached to each internal node of the path.}\label{figlowcatadv}
\end{figure}
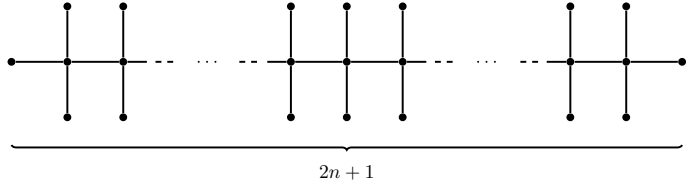

\begin{theorem}\label{lowcatadv}
Let $1/2 \leq \varepsilon < 1$. The advice complexity of any adaptive-priority algorithm for \textnormal{DPA} on trees with an approximation ratio smaller than $\frac{2}{1+\varepsilon}$ is at least $(1-\mathcal{H}(\varepsilon)) \left \lceil \frac{1}{2}\sum_{v \in V(T)} \left\lfloor \frac{\deg(v)}{4}\right\rfloor \right \rceil$.
\end{theorem}

\begin{proof}
Set $\sigma (T):= \left \lceil \frac{1}{2}\sum_{v \in V(T)} \left\lfloor \frac{\deg(v)}{4}\right\rfloor \right \rceil$ and let \textsc{ALG} be an adaptive-priority algorithm for \textnormal{DPA} on trees reading fewer than $(1-\mathcal{H}(\varepsilon))\sigma (T) $ advice bits. For each $n \in \N$, choose a tree $T_n$ with $\sigma(T_n)=n$. This is possible, as evidenced by the existence of the tree depicted in \cref{figlowcatadv}.

As in the proof of \cref{lowlwdpaadv}, we construct an online algorithm \textsc{TGUESS} for the \textnormal{2-SGKH} problem -- using \textsc{ALG} as a subroutine. If \textsc{TGUESS} has to guess a binary string of length $n$, it uses \textsc{ALG} on the tree $T_n$. Let $S^1,\,S^2,\,...\,,\,S^n$ be $n$ pairwise edge-disjoint copies of $S_4$ in $T_n$ and define $R_i=\{r_{i,\,1},\,r_{i,\,2},\,...\,,\,r_{i,\,6}\}$ to be the set of all paths of length $2$ in $S^i$.

\textsc{TGUESS} is conceptually identical to the algorithm \textsc{GUESS} defined in the proof of \cref{lowlwdpaadv}. It presents the request $r_{i,j}$ with highest priority in $R_i$ to \textsc{ALG} and guesses $1$ if and only if \textsc{ALG} accepts $r_{i,j}$. It then learns what the correct bit would have been and accordingly either presents to \textsc{ALG} the request in $R_i$ that does not intersect $r_{i,j}$, i.e., the request that consists of the other two edges in $S^i$, or it presents two non-intersecting request in $R^i$ that both intersect $r_{i,j}$. In the first case, accepting $r_{i,j}$ (and the second request) would have been optimal for \textsc{ALG}, while in the second case rejecting $r_{i,j}$ and accepting the other two requests would have been optimal. See Algorithm \ref{tguessalgo}.

\begin{algorithm}[h]

\caption{\textsc{TGUESS}}\label{tguessalgo}
\medskip

\textbf{Online Input: } The sequence $n,\,d_1,\, d_2,\, ...,\, d_n$, where $n\in \N$ and $d_i\in \{0,1\}$ for $i\in \{1,\,2,\,...\,,\,n\}$.

\medskip

\textbf{Advice: } The advice tape $\Phi$.

\medskip

\textbf{Online Output: } The sequence $y_1,\, y_2,\, ...,\, y_n$, where  $y_i\in \{0,1\}$ for $i\in \{1,\,2,\,...\,,\,n\}$.

\medskip

\settowidth{\characterlength}{(}

\textbf{Algorithm: }

\medskip

	\begin{algorithmic}
		\State Read the length $n\in \N$ of the binary string to guess.
		\State \hspace{-\characterlength}(Input the tree $T_n$ to \textsc{ALG}.) \Comment{Only informally.}
		\State $Q\gets \emptyset$
		\State $U\gets \{r_{i,j}\;|\;1\leq i\leq n,\;1\leq j\leq 6\}$
		\For {$1\leq k\leq n$}
			\While {$\max{(U\cup Q)} \in Q$}
				\State $m \gets \max{(U\cup Q)}$
				\State Feed $m$ to \textsc{ALG}
				\State $Q \gets Q\setminus \{m\}$
			\EndWhile
			\State $m_k\gets \max{(U\cup Q)}$ \Comment{$m_k\in U$.}
			\State Find $1 \leq i\leq n$ such that $m_k\in R_i$
       	\State $U \gets U \setminus R_i$
       	\State Feed $m_k$ to \textsc{ALG}, if \textsc{ALG} accepts, guess $y_k:=1$, else guess $y_k:=0$.
       	\State Read the true value $d_k$.
       	\If {$d_k=1$}
				\State $Q\gets Q\cup \{\text{the request in } R_i\text{ that does not intersect } m_k.\}$
			\EndIf
			\If {$d_k=0$}
				\State $Q\gets Q\cup (\text{two edge-disjoint requests in }R_i \text{ that intersect }m_k)$
			\EndIf
		\EndFor
		\While {$Q \neq \emptyset$} \Comment{Post-processing}
			\State $m \gets \max Q$
			\State Feed $m$ to \textsc{ALG}
			\State $Q \gets Q\setminus \{m\}$
		\EndWhile
	\end{algorithmic}
\end{algorithm}

Exactly as in \cref{lowlwdpaadv}, as a consequence of \cref{2sgkh} proved by Böckenhauer et al.~\cite{bockenhauer2014string}, \textsc{TGUESS} must guess at least $(1-\varepsilon) n$ bits incorrectly on some instance of \textnormal{2-SGKH} of size $n$, because it reads fewer than $(1-\mathcal{H}(\varepsilon))\sigma (T)= (1-\mathcal{H}(\varepsilon))n$ advice bits. Thus, for the corresponding instance $I$ of \textnormal{DPA} on paths, there are at least $(1-\varepsilon)n$ out of $n$ star trees $S^i$ on which \textsc{ALG} only accepts one request, while the optimal solution accepts two. This yields
\begin{equation*} 
\frac{|\textsc{OPT}(I)|}{| \textsc{ALG}(I)|}\geq\frac{2n}{(1-\varepsilon ) n + 2 \varepsilon n}=\frac{2}{1+\varepsilon}.\qedhere
\end{equation*}
\end{proof}

\section{Concluding Remarks}\label{sec:conclusion}
What has thus far remained untouched in this paper are priority algorithms for \textnormal{DPA} problems on graph classes containing graphs with cycles. In particular, the \textnormal{DPA} problem on grids -- another generalization of \textnormal{DPA} on paths -- has already been analyzed in the online framework by Böckenhauer et al.~\cite{bockenhauer2018call}. While their work shows that it is already quite challenging to obtain results for online algorithms, it is even more difficult to do so for priority algorithms. The complication lies in the fact that requests no longer correspond to unique paths, i.e., that an online or priority algorithm has multiple options for how to satisfy a request.

It is quite easy to show a lower bound of $3/2$ on the approximation ratio of any adaptive-priority algorithm without advice for \textnormal{DPA} on grids -- but only on a specific graph.

\begin{theorem}\label{thm:gridthm}
An adaptive-priority algorithm without advice for \textnormal{DPA} on grids cannot be better than strictly $3/2$-competitive on all grids.
\end{theorem}

The proof, which can be found in \cref{apdx:gridthm}, uses a construction on the $(3 \times 3)$-grid. We conjecture that this lower bound holds for all grids -- except, of course, for the $1 \times n$ grid, $n \in \N$, since this is simply the path of length $n-1$, on which \textnormal{DPA} is optimally solvable by \cref{thm:optdpa}. While this lower bound might be difficult to prove in the general case, proving it for specific grid sizes should be possible in a similar manner to how we prove it for the $(3 \times 3)$-grid. It is simply a matter of distinguishing between finitely many cases, which can be done by a computer.

\bibliographystyle{plainurl}
\bibliography{mybib}

\appendix
\newpage
\section{Proof of Theorem~\ref{thm:gridthm}}\label{apdx:gridthm}
\begin{proof}
Let \textsc{ALG} be an adaptive-priority algorithm without advice for \textnormal{DPA} on grids and let $G$ be a $(3 \times 3)$-grid as depicted below. Let $r$ be the request with highest initial priority among all requests in $G$ whose end points have distance $3$, where the distance between two vertices is the length of the shortest path connecting them. After suitably rotating $G$ and labeling its vertices, $r$ will be the request $[v,\,w]$ in \cref{figgrid}. As discussed previously in this paper, on every instance in which $r$ has highest initial priority, \textsc{ALG} will have to accept $r$ in order to achieve any approximation ratio at all. Depending on how \textsc{ALG} satisfies $r$, we consider different instances on which \textsc{ALG} will be bad compared to an optimal offline algorithm.

Say \textsc{ALG} satisfies $r$ by allocating the path $p$. If $p$ does not contain the center vertex of the grid, it must internally contain some corner vertex $c$. This allows us to define the instance $I_1$ consisting of $r$ and two requests containing the corner $c$ that both cannot be accepted by \textsc{ALG} but by the optimal solution, as depicted in \cref{figgrid2}, yielding $\abs{\textsc{OPT}(I_1)}/\abs{\textsc{ALG}(I_1)}=2$.

If $p$ does contain the center of the grid, the situation must be as in \cref{figgrid3} (after suitably rotating and reflecting $G$). This allows us to define $I_2$ as the instance consisting of $r$ and three additional requests that are all contained in the optimal solution but of which at most one can be accepted by \textsc{ALG}. This yields $\abs{\textsc{OPT}(I_2)}/\abs{\textsc{ALG}(I_2)} \geq 3/2$.
\end{proof}
\newpage
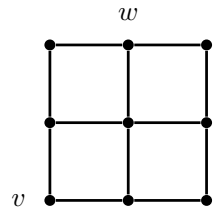
\begin{figure}
\centering
\resizebox{.22\textwidth}{!}{

\begin{tikzpicture}
\node[circle, fill=black, inner sep=0 , minimum size=4] at (0,0) (node00) {};
\node[circle,fill=black, inner sep=0 , minimum size=4] at (0,1) (node01) {};
\node[circle,fill=black, inner sep=0 , minimum size=4] at (0,2) (node02) {};

\foreach \n in {1,...,2}{
		\foreach \m in {0,...,2}{
		\pgfmathtruncatemacro{\nminusone}{\n - 1}
       \node[circle,fill=black, inner sep=0 , minimum size=4] at (\n,\m) (node\n\m) {};
       \draw[line width=1] (node\nminusone\m)--(node\n\m);
		}
    }
\foreach \n in {0,...,2}{
	\foreach \m in {1,...,2}{
	\pgfmathtruncatemacro{\mminusone}{\m - 1}
   	\draw[line width=1] (node\n\mminusone)--(node\n\m);
		}
    }
\node[circle] at (-0.4,-0){$v$};
\node[circle] at (1, 2.4) {$w$};

\end{tikzpicture}

}

\caption{The $(3 \times 3)$-grid $G$ and the request $r=[v,\,w]:=\max_{\prec}\{[a,\,b]\;|\; d(a,\,b)=3\}$, where $\prec$ is the initial priority order. Note that $r$ must be accepted by \textsc{ALG} when operating on any instance beginning with $r$, because otherwise, \textsc{ALG} has a gain of $0$ on the instance consisting only of $r$, while the optimal solution has a gain of $1$.}
\label{figgrid}
\end{figure}

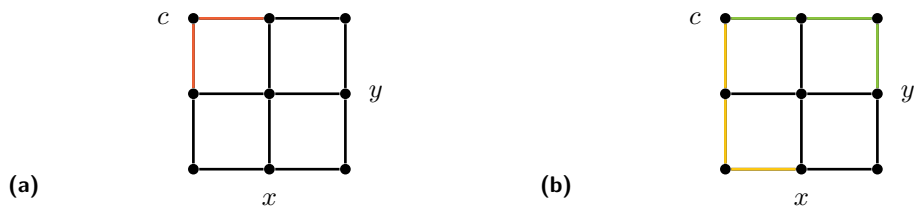
\begin{figure}
\centering
\begin{subfigure}{.5\textwidth}
\centering
\begin{tikzpicture}

\node[circle,fill=black, inner sep=0, minimum size=4] at (0,0) (node00) {};
\node[circle,fill=black, inner sep=0, minimum size=4] at (0,1) (node01) {};
\node[circle,fill=black, inner sep=0, minimum size=4] at (0,2) (node02) {};

\foreach \n in {1,...,2}{
		\foreach \m in {0,...,2}{
		\pgfmathtruncatemacro{\nminusone}{\n - 1}
       \node[circle,fill=black, inner sep=0, minimum size=4] at (\n,\m) (node\n\m) {};
       \draw[line width=1] (node\nminusone\m)--(node\n\m);
		}
    }
\foreach \n in {0,...,2}{
	\foreach \m in {1,...,2}{
	\pgfmathtruncatemacro{\mminusone}{\m - 1}
   	\draw[line width=1] (node\n\mminusone)--(node\n\m);
		}
    }
\node[circle] at (-0.4,2){$c$};
\node[circle] at (1, -0.4) {$x$};
\node[circle] at (2.4, 1) {$y$};

\draw[line width=1, color=RedOrange] (node02)--(node12);
\draw[line width=1, color=RedOrange] (node02)--(node01);

\end{tikzpicture}
\vspace*{-5.5ex}
\caption{}
\label{figgrid2a}
\end{subfigure}%
\begin{subfigure}{.5\textwidth}
\centering
\begin{tikzpicture}

\node[circle,fill=black, inner sep=0, minimum size=4] at (0,0) (node00) {};
\node[circle,fill=black, inner sep=0, minimum size=4] at (0,1) (node01) {};
\node[circle,fill=black, inner sep=0, minimum size=4] at (0,2) (node02) {};

\foreach \n in {1,...,2}{
		\foreach \m in {0,...,2}{
		\pgfmathtruncatemacro{\nminusone}{\n - 1}
       \node[circle,fill=black, inner sep=0, minimum size=4] at (\n,\m) (node\n\m) {};
       \draw[line width=1] (node\nminusone\m)--(node\n\m);
		}
    }
\foreach \n in {0,...,2}{
	\foreach \m in {1,...,2}{
	\pgfmathtruncatemacro{\mminusone}{\m - 1}
   	\draw[line width=1] (node\n\mminusone)--(node\n\m);
		}
    }

\node[circle] at (-0.4,2){$c$};
\node[circle] at (1, -0.4) {$x$};
\node[circle] at (2.4, 1) {$y$};

\draw[line width=1, color=lipicsYellow] (node02)--(node01);
\draw[line width=1, color=lipicsYellow] (node01)--(node00);
\draw[line width=1, color=lipicsYellow] (node00)--(node10);

\draw[line width=1, color=LimeGreen] (node02)--(node12);
\draw[line width=1, color=LimeGreen] (node12)--(node22);
\draw[line width=1, color=LimeGreen] (node22)--(node21);

\end{tikzpicture}
\vspace*{-5.5ex}
\caption{}
\label{figgrid2b}
\end{subfigure}
\caption{If $p$ does not contain the center of $G$, it must traverse some corner vertex $c$ to get from $v$ to $w$, i.e., it must contain the two red edges incident to $c$ in Subfigure (a) (after suitably rotating the grid). This means that both the requests $[c,\,x]$ and $[c,\,y]$ -- with $x$ and $y$ as in the figures -- cannot be accepted by \textsc{ALG} because both edges incident to $c$ are blocked (note that $r$ has higher priority than $[c,\,x]$ and $[c,\,y]$ by definition). The optimal solution, however, rejects $r$ and accepts the other two requests, e.g., by allocating the yellow and the green path in Subfigure (b). This yields $\abs{\textsc{OPT}(I_1)}/\abs{\textsc{ALG}(I_1)} = 2$.}
\label{figgrid2}
\end{figure}
\begin{figure}[h]
\centering
\begin{subfigure}{.5\textwidth}
\centering
\begin{tikzpicture}

\node[circle,fill=black, inner sep=0, minimum size=4] at (0,0) (node00) {};
\node[circle,fill=black, inner sep=0, minimum size=4] at (0,1) (node01) {};
\node[circle,fill=black, inner sep=0, minimum size=4] at (0,2) (node02) {};

\foreach \n in {1,...,2}{
		\foreach \m in {0,...,2}{
		\pgfmathtruncatemacro{\nminusone}{\n - 1}
       \node[circle,fill=black, inner sep=0, minimum size=4] at (\n,\m) (node\n\m) {};
		}
    }
\foreach \n in {0,...,2}{
	\foreach \m in {1,...,2}{
	\pgfmathtruncatemacro{\mminusone}{\m - 1}
   	\draw[line width=1] (node\n\mminusone)--(node\n\m);
		}
    }

\draw[line width=1] (node00)--(node10);
\draw[line width=1] (node10)--(node20);
\draw[line width=1] (node01)--(node11);
\draw[line width=1] (node11)--(node21);
\draw[dashed, line width=1] (node02)--(node12);
\draw[line width=1] (node12)--(node22);

\node[circle] at (-0.4,1){$t$};
\node[circle] at (-0.4, 2) {$t'$};
\node[circle] at (1, -0.4) {$x$};
\node[circle] at (2.4, 2) {$z$};
\node[circle] at (2.4, -0.4) {$y$};

\draw[line width=1, color=RedOrange] (node00)--(node01);
\draw[line width=1, color=RedOrange] (node01)--(node11);

\end{tikzpicture}
\vspace*{-5.5ex}
\caption{}
\label{figgrid3a}
\end{subfigure}%
\begin{subfigure}{.5\textwidth}
\centering
\begin{tikzpicture}

\node[circle,fill=black, inner sep=0, minimum size=4] at (0,0) (node00) {};
\node[circle,fill=black, inner sep=0, minimum size=4] at (0,1) (node01) {};
\node[circle,fill=black, inner sep=0, minimum size=4] at (0,2) (node02) {};

\foreach \n in {1,...,2}{
		\foreach \m in {0,...,2}{
		\pgfmathtruncatemacro{\nminusone}{\n - 1}
       \node[circle,fill=black, inner sep=0, minimum size=4] at (\n,\m) (node\n\m) {};
       \draw[line width=1] (node\nminusone\m)--(node\n\m);
		}
    }
\foreach \n in {0,...,2}{
	\foreach \m in {1,...,2}{
	\pgfmathtruncatemacro{\mminusone}{\m - 1}
   	\draw[line width=1] (node\n\mminusone)--(node\n\m);
		}
    }

\node[circle] at (-0.4,1){$t$};
\node[circle] at (-0.4, 2) {$t'$};
\node[circle] at (1, -0.4) {$x$};
\node[circle] at (2.4, 2) {$z$};
\node[circle] at (2.4, -0.4) {$y$};

\draw[line width=1, color=LimeGreen] (node01)--(node11);
\draw[line width=1, color=LimeGreen] (node11)--(node21);
\draw[line width=1, color=LimeGreen] (node21)--(node22);

\draw[line width=1, color=lipicsYellow] (node01)--(node00);
\draw[line width=1, color=lipicsYellow] (node00)--(node10);
\draw[line width=1, color=lipicsYellow] (node10)--(node20);

\draw[line width=1, color=NavyBlue] (node02)--(node12);
\draw[line width=1, color=NavyBlue] (node12)--(node11);
\draw[line width=1, color=NavyBlue] (node11)--(node10);
\end{tikzpicture}
\vspace*{-5.5ex}
\caption{}
\label{figgrid3b}
\end{subfigure}
\caption{If $p$ contains the center of the grid, it needs to traverse at least $2$ edges before getting there from $v$, and must thus contain the two red edges incident to the vertex $t$ in Subfigure (a) (after suitably rotating and reflecting the grid). The additional requests $[t,\,y]$, $[t,\,z]$ and $[t',\,x]$ -- with $x$, $y$, and $z$ as in the figures -- all have lower priority than $r$ by definition. \textsc{ALG} can accept at most one of these additional requests since an allocated path for any of them must use the dashed edge. The optimal solution rejects $r$ and accepts the other three requests, e.g., by allocating the yellow, the green and the blue path as in Subfigure (b). This yields $\abs{\textsc{OPT}(I_2)}/\abs{\textsc{ALG}(I_2)} \geq 3/2$.}
\label{figgrid3}
\vspace*{-9ex}
\end{figure}
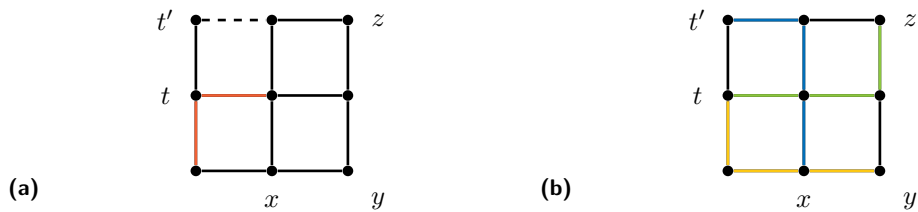

\end{document}